\documentclass[10pt]{article}
\usepackage[applemac]{inputenc}
\usepackage{amsmath}
\usepackage{amsfonts}
\usepackage{amssymb}
\usepackage{amsthm}
\usepackage{graphicx}
\usepackage{stmaryrd}
\usepackage{pdfsync}
\usepackage[top=3cm, bottom=3cm, left=3cm, right=3cm]{geometry} 
\usepackage{bm}

\usepackage{color}
\usepackage{soul,xcolor}

\def\R {\mathbb{R}}

\def\coth {\text{\rm coth}\,}
\def\csch {\text{\rm csch}\,}

\newtheorem{theo}{Theorem}

\newtheorem{prop}{Proposition}
\newtheorem{cor}{Corollary}

\newtheorem*{exe}{Example}

\newcommand{\secta}[1]{\section*{#1}\setcounter{section}{1}
                                    \setcounter{equation}{0}
\def\theequation{\Alph{section}.\arabic{equation}}}

\newcommand{\sectb}[1]{\section*{#1}\addtocounter{section}{1}
                                    \setcounter{equation}{0}
\def\theequation{\Alph{section}.\arabic{equation}}}

\newcommand{\sectc}[1]{\section*{#1}\addtocounter{section}{1}
                                    \setcounter{equation}{0}
\def\theequation{\Alph{section}.\arabic{equation}}}

\newcommand{\sectd}[1]{\section*{#1}\addtocounter{section}{1}
                                    \setcounter{equation}{0}
\def\theequation{\Alph{section}.\arabic{equation}}}

\newcommand{\secte}[1]{\section*{#1}\addtocounter{section}{1}
                                    \setcounter{equation}{0}
\def\theequation{\Alph{section}.\arabic{equation}}}

\newlength{\defbaselineskip}
\newcommand{\setlinespacing}[1]
    {\setlength{\baselineskip}{#1 \defbaselineskip}}
\setlinespacing{1.5}

\begin{document}

%\title{Free surface effects on the stability of two-layer flows with piecewise linear shear current}
\title{Effect of variation in density on the stability of bilinear shear currents with a free surface}
%\title{Effect of variation in density on the stability of piecewise linear shear currents with a free surface}
\author{Ricardo Barros$^{1,a)}$ and Jos\'{e} Felipe Voloch$^{2}$\\[6pt]     
\small $^1$ Department of Mathematical Sciences\\
\small Loughborough University\\
\small Loughborough LE11 3TU, UK\\
\small a) Electronic mail: r.barros@lboro.ac.uk\\ 
%Tel.: (+44) 1509228256.\\ 
\\
\small $^2$ School of Mathematics and Statistics\\
\small University of Canterbury\\
\small Private Bag 4800, Christchurch 8140, New Zealand
}

%\author{Ricardo Barros$^1$ and Wooyoung Choi$^{2}$\\[6pt]     
%\small $^1$ Mathematics Applications Consortium for Science and Industry (MACSI)\\
%\small Department of Mathematics and Statistics\\
%\small University of Limerick\\
%\small Limerick, Ireland\\
%\\
%\small $^2$Department of Mathematical Sciences\\
%\small New Jersey Institute of Technology\\
%\small Newark, NJ 07102-1982, USA\\
%}

\date{}

\maketitle

\begin{abstract}
We perform the stability analysis for a free surface fluid current modeled as two finite layers of constant vorticity, under the action of gravity and absence of surface tension. 
%where the influence of surface tension is neglected. 
In the same spirit as Taylor [``Effect of variation in density on the stability of superposed streams of fluid," Proc.~R.~Soc.~A {\bf 132}, 499 (1931)], a geometrical approach to the problem is proposed, which allows us to present simple analytical criteria under which the flow is stable. 
%for the stability of the flow. se lowand derive a necessary and sufficient condition under which the flow is stable. 
A strong destabilizing effect of stratification in density is perceived when the results are compared with those obtained for the 
%revealed 
%by comparison with the 
physical setting where the vorticity interface is also a density interface separating two immiscible fluids with constant densities. 
In contrast with the homogenous case, the stratified bilinear shear current is mostly unstable and can only be stabilized when the background current in the upper layer is constant.
%Comparison with the physical setting where the vorticity interface is also a density interface reveals that stratification in density has a strong destabilizing effect.   
%criterion for stability of the flow.  
\end{abstract}
%\section{Linear stability analysis}

\section{Introduction}
There has been an increased interest in examining the vorticity effects on nonlinear water waves since the contributions of Benjamin \cite{Benjamin} and Benney \cite{Benney}, based on weakly nonlinear theory. Although the need to account for strong nonlinearity in the description of surface gravity waves, especially when the shear current is strong, has long been recognized \cite{Teles_da_Silva_Peregrine}, little analytical progress has been done beyond the weakly nonlinear regime until rather recently. 

For finite-depth uniform shear currents, Choi \cite{Choi_2003} proposed a strongly nonlinear long wave model that captures the salient features revealed by the numerical studies for the Euler equations~\cite{Teles_da_Silva_Peregrine, Vanden_Broeck}. Namely, the fact that shear currents modify the shape of solitary waves, leading possibly to overturning waves, and the appearance of stationary recirculating eddies. Rigorous results on the existence of steady periodic waves for Euler equations with general vorticity distribution have since appeared in the literature, subsequent to the contributions by Constantin and Strauss \cite{Constantin_Strauss_2004, Constantin_Strauss_2011}.

Density stratified shear flows have likewise received considerable attention over the years. As for homogeneous fluids, most progress made beyond the weakly nonlinear theory developed by Benjamin~\cite{Benjamin_1966} and Grimshaw~\cite{Grimshaw} was, until recently, based on numerical studies for Euler equations. Modeling a stratified shear flow by approximating the velocity and density profiles by piecewise-linear and piecewise-constant functions, respectively, is a tempting and common approach due to its mathematical tractability. In particular, 
steady large-amplitude waves in a two-layer density-stratified flow with a bilinear shear current were extensively studied by Pullin and Grimshaw~\cite{Pullin_Grimshaw_1983,Pullin_Grimshaw_1986}, 
%
%the case of a two-layer density-stratified flow with a bilinear shear current was extensively studied by Pullin and Grimshaw~\cite{Pullin_Grimshaw_1983,Pullin_Grimshaw_1986}, 
by assuming the fluid domain to be bounded by rigid boundaries. The influence of having in system the top rigid boundary replaced by a free surface was recently investigated by Curtis, Oliveras, and Morrison \cite{Curtis_et_al}.
%who focused on the characteristics of steady large-amplitude solutions for this physical setting. 
%The mathematical tractability of the piecewise-linear velocity profile makes it attractive to work with.

In the aforementioned studies, the question of the stability of the background flow often lies unaddressed. As we know, hydrodynamic stability is an active field of research and stability of a flow may depend in a rather subtle way on the details of the velocity and density profiles. Stratification in density, often assumed to have a stabilizing effect, can also promote instability as early pointed out by Goldstein~\cite{Goldstein} and Taylor~\cite{Taylor}. It can also alter qualitatively the character of instability of shear flows, as clearly demonstrated by Holmboe~\cite{Holmboe}. 
Although instability can easily be established for all the classical configurations proposed by Taylor \cite{Taylor}, Goldstein \cite{Goldstein}, and Holmboe \cite{Holmboe}, the mathematics involved do not clarify the physical mechanisms leading to the shear instability. Taylor~\cite{Taylor} has proposed that instability in the so-called {\it Taylor-Goldstein configuration} could be a result of a resonant interaction between two interfacial gravity waves. The wave interaction interpretation of instability was greatly developed by many authors, including Baines \& Mitsudera~\cite{Baines_Mitsudera} and Caulfield~\cite{Caulfield}, and contributed to the understanding of some main results from the stability of stratified shear flows (see the excellent review paper by Carpenter, Tedford, Heifetz, and Lawrence~\cite{Carpenter_et_al} and references therein). However, boundary effects 
%and non-Boussinesq effects 
are often neglected in these descriptions, which may significantly affect the stability of the flow.    

For homogeneous fluids, attempts to extend Rayleigh's inflexion Theorem \cite{Drazin_Reid} to the free surface setting 
%where the upper rigid boundary is replaced by a free surface 
were made by Yih~\cite{Yih} and, more recently, Hur and Lin \cite{Hur_Lin} (see also the related paper by McHugh~\cite{McHugh}). Yih has stated that a monotonic profile for the background velocity with no inflexion points is stable. Hur and Lin claim to have extended Yih's result further (by relaxing the monotonicity assumption on the mean velocity), but were promptly disproved by Renardy and Renardy~\cite{Renardy_Renardy}. These works highlight, in particular, how sensitive the stability of a flow may be to the upper boundary condition. 

In stratified flows, the confinement by rigid boundaries and neglect of the Boussinesq approximation can also lead to surprising effects on the stability of the flow. Namely, all the classical configurations proposed by Taylor~\cite{Taylor}, Goldstein~\cite{Goldstein}, and Holmboe~\cite{Holmboe} can be stabilized at low Richardson number, as recently shown by Barros and Choi~\cite{Barros_Choi_2014}.

Here we consider the stability analysis of a free surface flow composed of two finite layers, each with a constant density and linear shear current. In this study, the mean horizontal velocity is assumed to be continuous, and viscosity and surface tension effects are neglected. To better understand the role of variation in density on the stability of flow, all results are compared in parallel with the corresponding homogenous case of a bilinear shear current. Following the steps of Taylor~\cite{Taylor}, and using some of the results in Barros and Choi~\cite{Barros_Choi_2014}, the analysis of the problem is performed with recourse to the theory of plane algebraic curves. One of the most pertinent aspects of the present work is that the geometrical approach leads to simple analytical criteria that are necessary and sufficient for the stability of the flow.
While the results are mathematically rigorous, no attempt is made to clarify in physical terms the reason why they are true.

%With recourse to the theory of plane algebraic curves, as used previously by Taylor~\cite{Taylor}, explicit criteria for stability are obtained. 
%%Stratification in density is shown to have a strong destabilizing effect. 
%All criteria for stability proposed in the present work are necessary and sufficient.  

%
%The aim of this work is to be able to describe 
%free surface fluid current modeled as two finite layers of constant vorticity, under the action of gravity and absence of surface tension. 

%In the case of small density variations for which the Boussinesq approx- imation has been often adopted, it has been shown13 that great care should be taken to ensure validity of a simplified analysis based on such approximation. Therefore, the classical Boussinesq approximation is also questioned in capturing the stability feature of low Richardson number flows.

%TALK ABOUT THE SEVERAL EXAMPLES WHICH HAVE EXTENSIVELY STUDIED FOR HOMOGENEOUS FLUIDS
%This particular setting has been considered by Thomson (1949).

\section{Formulation}\label{sec:formulation}

The stability of an inviscid, incompressible, stratified shear flow depends upon the vertical variation of density $\rho(z)$ and the mean horizontal velocity $U(z)$. The behavior of a small two-dimensional, monochromatic disturbance of wavenumber $k$ and wave speed $c$ is governed (see {\it e.g.} \cite{Drazin_Reid}) by
\begin{equation}\label{model}
\phi'' +\frac{\rho'}{\rho}\left (\phi'-{U'\over U-c}\phi\right )+ \left[-\frac{g\,\rho'}{\rho\, (U-c)^2}  - \frac{U''}{U-c} - k^2 \right] \phi=0,
\end{equation}
where the prime indicates differentiation with respect to $z$, $g$ is the gravitational acceleration, and $\phi$ is the complex amplitude of the stream function 
$\psi$ defined by $\psi (x,z,t) = \phi(z) \exp [i k(x-ct)]$ at each point $(x,z)$ and time $t$. The wave speed $c$ may be complex, and such wave is said to be unstable if ${\rm Im}(c)>0$. 
%\tcb{More precisely, if $\psi (x,z,t)$ represents the infinitesimal deviation from the values of the stream function in the undisturbed state, then it is given by $\psi (x,z,t) = \phi(z) \exp [i k(x-ct)]$. The wave speed $c$ may be complex, and such wave is said to be unstable if ${\rm Im}(c)>0$.} 
%{\it i.e.,} $c=c_r+ic_i$, and such a wave is said to be unstable if $c_i>0$.

In this study, piecewise linear velocity and piecewise constant density profiles are adopted (see Fig.~\ref{sketch}$(a)$):
\begin{equation}\label{physical_configuration}
U(z) = \left\{ 
\begin{array}{ccl}
u_0+\Omega_2 z & \mbox{if } & 0<z<H_2\\
u_0 +\Omega_1 z & \mbox{if } & -H_1<z < 0\\
\end{array}  
\right.
, \quad 
\rho(z) = \left\{ 
\begin{array}{rcl}
\rho_2 & \mbox{if } & 0<z<H_2\\
\rho_1 & \mbox{if } & -H_1<z < 0\\
\end{array}  
\right. .
\end{equation}
If the flow were to be confined between two rigid walls, it is well known that the flow is stable regardless of the physical parameters. The linear dispersion relation is given by a quadratic equation for the wave speed $c$, which has always two real roots \cite{Choi_2006}. Here we investigate the effects on the stability features of the flow caused by the presence of a top free surface.  
Furthermore, to isolate the influence of variation in density on the stability of the flow, comparison with the results for a homogeneous fluid (see Fig.~\ref{sketch}$(b)$) will be presented.

In each subdomain $-H_1<z<0$ or $0<z<H_2$, equation \eqref{model} can be solved explicitly as:
%as linear combinations of $\exp (\pm kz)$
$$
\phi(z) = \left\{
\begin{array}{ccl}
A_2 \,e^{kz} + B_2 \,e^{-kz} & \mbox{if } & 0<z<H_2\\
A_1 \,e^{kz} + B_1\, e^{-kz} & \mbox{if } & -H_1<z<0\\
\end{array}
\right. ,
$$
for arbitrary constants $A_1, A_2, B_1, B_2$.
Then, at the level $z=0$, where $\rho (z)$ is discontinuous, the continuity of pressure and normal velocity at this  surface requires the following jump conditions
\begin{equation}\label{1st_jump_condition}
\left\llbracket  \rho \left[ (U-c) \, \phi^\prime - \left( U^\prime + \frac{g}{U-c} \right) \phi \right]  \right\rrbracket =0,\qquad
\left\llbracket \phi \right\rrbracket =0,
\end{equation}
respectively. Here we have used $\llbracket \cdot \rrbracket$ to denote a jump across the interface. Along with these, the following boundary conditions are imposed, respectively, at the top and bottom of the flow domain:
\begin{equation}\label{boundary_conditions}
\phi^\prime (H_2)- \left[ \frac{g}{(U(H_2)-c)^2} + \frac{U^\prime(H_2)}{U(H_2)-c}\right] \phi(H_2)=0, \quad \phi(-H_1)=0. 
\end{equation}

\begin{figure}[htbp]
\begin{center}
\includegraphics*[width=430pt]{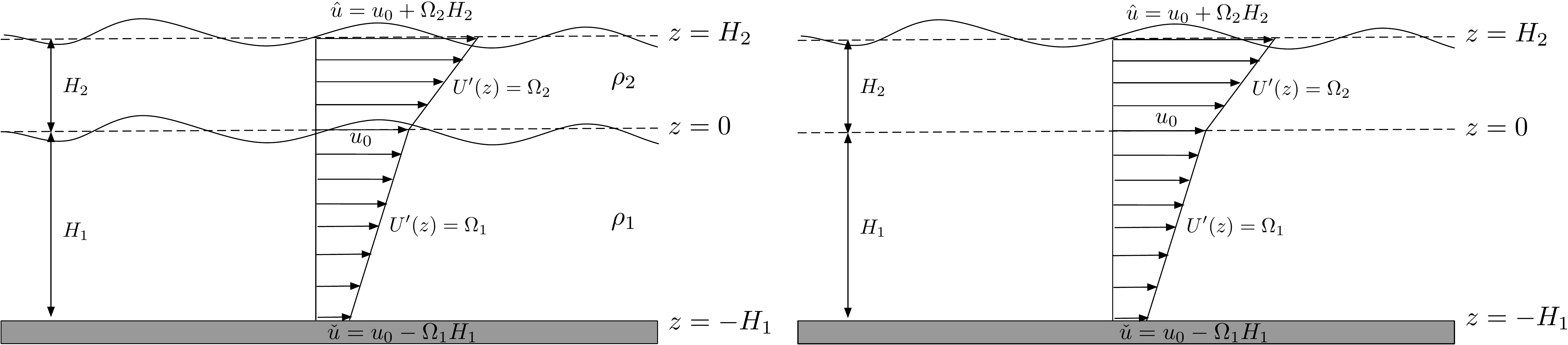}
\\
$(a)$ \hspace{4cm} $(b)$
\end{center}
\caption{Schematic of the physical configurations examined in this paper. Panel $(a)$ illustrates the stratified bilinear shear current defined in \eqref{physical_configuration}. The corresponding homogeneous bilinear shear current ($\rho_2 = \rho_1$) is illustrated in panel $(b)$.
\label{sketch}
}
\end{figure}

%\begin{figure}[htbp]
%\begin{center}
%\includegraphics*[width=215pt]{figures/sketch_stratified_bilinear_shear_current}
%\quad 
%\includegraphics*[width=215pt]{figures/sketch_bilinear_shear_current}
%\\
%$(a)$ \hspace{4cm} $(b)$
%\end{center}
%\caption{\tcb{Schematic of the physical configurations examined in this paper. Panel $(a)$ illustrates the stratified bilinear shear current defined in \eqref{physical_configuration}. The corresponding homogeneous bilinear shear current ($\rho_2 = \rho_1$) is illustrated in panel $(b)$.}
%\label{sketch}
%}
%\end{figure}

Equations \eqref{1st_jump_condition}--\eqref{boundary_conditions} yield a linear system composed by 4 equations for 4 unknowns. Following Taylor (see pp.~509-511 in \cite{Taylor}), the system can be easily reduced to two equations: 
$$
- k \,\csch(kH_2) (\hat{u}-c)^2 \,c_1 + \left[ k \,\coth(kH_2) (\hat{u}-c)^2 - \Omega_2 (\hat{u}-c)-g \right] c_2=0, 
$$
\begin{multline}
\Big\{\rho_2 \left[ k \,\coth(kH_2) (u_0-c)^2+ \Omega_2 (u_0-c)+g\right] + \rho_1 \left[ k \,\coth(kH_1) (u_0-c)^2 - \Omega_1 (u_0-c)-g\right] \Big\} c_1 - \\
-\rho_2 \,k \,\csch(kH_2) (u_0-c)^2 \, c_2=0,
\end{multline}
%located at $z=\pm H$
%we can readily obtain the dispersion relation 
with $c_1 \equiv A_2 + B_2$ and $c_2 \equiv A_2 \,e^{kH_2}+B_2 \,e^{-kH_2}$,
from which it follows that the dispersion relation between the wave speed $c$ and the wavenumber $k$ is obtained as a polynomial equation (of degree 4) for $c$:
\begin{multline}\label{dispersion_relation}
\Big\{\rho_1 \left[ k\,\coth(kH_1) \,(u_0-c)^2 - \Omega_1 (u_0-c)-g \right] + \rho_2 \left[ k \,\coth(kH_2) \,(u_0-c)^2+ \Omega_2 (u_0-c)+g \right] \Big\}  \times \\
\times \left[ k \,\coth(kH_2) \,(\hat{u}-c)^2 - \Omega_2 (\hat{u}-c)-g \right] =\rho_2 \,k^2 \,\csch^2(kH_2) \,(u_0-c)^2 (\hat{u}-c)^2.
\end{multline}
Here, $u_0$ and $\hat{u}$ are simply the velocities at $z=0$ and $z=H_2$, respectively. 
While formulating the problem is straightforward, fully describing the stability features of the flow is rather difficult in practice due to the large number of physical parameters in the problem. 
%so complex that it is rather difficult to fully describe stability features of these flows. 
To avoid this difficulty, following Taylor~\cite{Taylor}, Chandrasekhar (see \S 105 in \cite{Chandrasekhar}), and Ovsyannikov~\cite{Ovsyannikov}, we adopt a geometrical approach based on the theory of plane algebraic curves. With introducing variables $p$ and $q$ defined by 
\begin{equation}\label{p_q_definition}
p=(u_0-c)/\sqrt{gH_1}, \quad q=(\hat{u}-c)/\sqrt{gH_1},
\end{equation}
the eigenvalue equation \eqref{dispersion_relation} for $c$ can then be written as:
%the eigenvalue equation \eqref{dispersion_relation} for $c$ can then be written as an algebraic curve $P(p,q)=0$ of degree 4 on the $(p,q)$-plane:
\begin{equation}\label{curve}
\left[ H(\beta_1 \, p^2 -\widetilde{\Omega}_1 p-1)+ \rho \left( \beta_2 \,p^2+ p(q-p)+H \right) \right]  \left( \beta_2 \,q^2 - q(q-p)-H \right)  = \rho  \beta_3 \, p^2 q^2.
\end{equation}
The coefficients $\beta_i$ in the expression are given by
$$
\beta_1 = \alpha \, \coth \alpha, \quad \beta_2 = H \alpha \,\coth (H \alpha), \quad \beta_3 = H^2 \alpha^2 \,\csch^2 (H \alpha),
$$
where $\alpha=kH_1$ is the dimensionless wavenumber, $\rho=\rho_2/\rho_1$ is the density ratio, $H=H_2/H_1$ is the depth ratio, and  $\widetilde{\Omega}_1 = \Omega_1 \sqrt{H_1/g}$ is the dimensionless vorticity in the lower layer. So long as $p$ and $q$ are real, equation \eqref{curve} defines an algebraic curve of degree 4 in the $(p,q)$-plane that we will designate by $P(p,q)=0$.
Hereafter, we assume that $\alpha$ and $H$ are strictly positive, and $0<\rho<1$, unless clearly stated otherwise, under which $\beta_1>1$, $\beta_2>1$, and $0<\beta_3 < 1$.
%which will be referred to as ``in general'' when making an assertion. Under these assumptions, $\beta_1>1$, $\beta_2>1$, and $0<\beta_3\leqslant 1$.

We remark that in \eqref{curve} the vorticity in the upper layer does not appear as a parameter. This is achieved by noticing that 
$$
\Omega_2 = [(\hat{u}-c)- (u_0-c)]/H_2,
$$
or  
\begin{equation}\label{Omega2_of_p_q}
\Omega_2 = (q-p) \frac{\sqrt{g H_1}}{H_2},
\end{equation}
which can be directly inserted into \eqref{dispersion_relation}, along with \eqref{p_q_definition}, to yield \eqref{curve}.
For a given difference of the velocities $\hat{u}$ and $u_0$, it follows from \eqref{p_q_definition} that $p$ and $q$ must lie on the straight line
\begin{equation}\label{relationship_p_q}
q = p + F,
\end{equation}
where $F$ stands for the Froude number defined as $F= (\hat{u}-u_0)/\sqrt{gH_1}$. The Froude number is related to the vorticity in the upper layer through $F= H \widetilde{\Omega}_2$, with $\widetilde{\Omega}_2 = \Omega_2 \sqrt{H_1/g}$. As such, large (small) values of $F$ can be thought of as large (small) values of $\widetilde{\Omega}_2$.
To find whether stable modes are possible for a given value of $F$, we need only to determine whether the corresponding line \eqref{relationship_p_q} has any intersections with the real part of the $(p,q)$-locus described by \eqref{curve}, which we have designated by the plane algebraic curve $P(p,q)=0$.
It follows from definition \eqref{p_q_definition} that each real root of Eq.~\eqref{dispersion_relation} will correspond to an intersection point between the curve $P(p,q)=0$ and the line of equation \eqref{relationship_p_q}.
%Notice that from definition \eqref{p_q_definition}, \tcb{$p$ and $q$ are real} if and only if the wave speed $c$ is real. 
%\tcb{As such, each real root of Eq.~\eqref{dispersion_relation} will correspond to an intersection point between the curve $P(p,q)=0$ and the line of equation \eqref{relationship_p_q}. 
Given that \eqref{dispersion_relation} is a quartic equation for $c$,
%Therefore, 
the flow is stable if the line \eqref{relationship_p_q} intersects curve \eqref{curve} at four points, and unstable otherwise.
This algebraic approach allows us to easily enumerate the real wave speeds for the eigenvalue problem. For convenience, hereafter, when referring to a straight line or an algebraic curve (of degree greater than 1), we will adopt the terminology {\it line} and {\it curve}, respectively.

To compare our results with those obtained for a homogeneous fluid, it suffices to consider in the analysis the limit when $\rho \rightarrow 1$. The curve \eqref{curve} becomes degenerate and splits up into a vertical line $p=0$ and a cubic curve. The vertical line is clearly a spurious solution, resulting from the fact that there is no longer a density interface, but simply a vorticity interface. The study reduces then to the analysis of an algebraic curve of degree 3. If, when intersected by the line \eqref{relationship_p_q}, three points are obtained, the homogeneous flow is stable. Otherwise, it is unstable.

As we will see, some departures from Eq.~\eqref{curve} will be made throughout the text. 
%The reason being that a complete study of this family of curves is rather difficult. 
While the particular form in \eqref{curve} is convenient to establish the stability of long waves (see \ref{sec:longwave_limit}) and rule out the existence of four complex solutions to the eigenvalue equation \eqref{dispersion_relation} (see Appendix A), classifying its singular points, key to distinguish the possible configurations for the family of curves, is impractical. Equivalent formulations will be proposed with certain symmetry properties that will allow us to overcome this issue.

 %, under which symmetry about the origin is obtained,  will be proposed to overcome this issue. 

%The geometrical approach presented here is, of course, equivalent to the standard analytical approach where a quartic equation for the phase velocity $c$, or equivalently for $p$, is obtained by substituting \eqref{relationship_p_q} into \eqref{curve}. Then, using Fuller's root location criteria (Fuller 1981; Jury \& Mansour 1981), stability diagrams on the $(\alpha,F)$-plane can be drawn (see {\it e.g.} Figure~\ref{stability_diagram_irrotational_lower_layer}). Typically, three colors would be required to distinguish a stable region from the two kinds of unstable regions (characterized by two, or four complex roots for the eigenvalue equation). In all diagrams presented here, however, two colors suffice, the reason being that four complex solutions to the eigenvalue equation \eqref{dispersion_relation} can never exist (see details in Appendix).

\section{Some stability results}

We start by considering a few examples where the stability analysis can 
%Here, we consider a few cases where the analysis for the linear stability problem can 
be considerably simplified. First, we examine the long wave limit when $\alpha \rightarrow 0$, followed by the case when the background current in the upper layer is constant.
%the upper layer \tcb{has a constant background current}. 
In both cases the algebraic curve describing the linear stability problem becomes the graph of a real-valued function and stability is readily established. Second,  we examine the cases of a constant background current in the lower layer, and the case of uniform vorticity. In the two cases, algebraic curves, symmetric about the origin, will be proposed to describe the linear stability problem. As we will see, fully classifying the curve configurations becomes accessible, and instability is proved to hold at least for a finite range of Froude numbers.

\subsection{Stable configurations}

\subsubsection{Long wave limit}\label{sec:longwave_limit}

\begin{figure}[htbp]
\begin{center}
\includegraphics*[width=340pt]{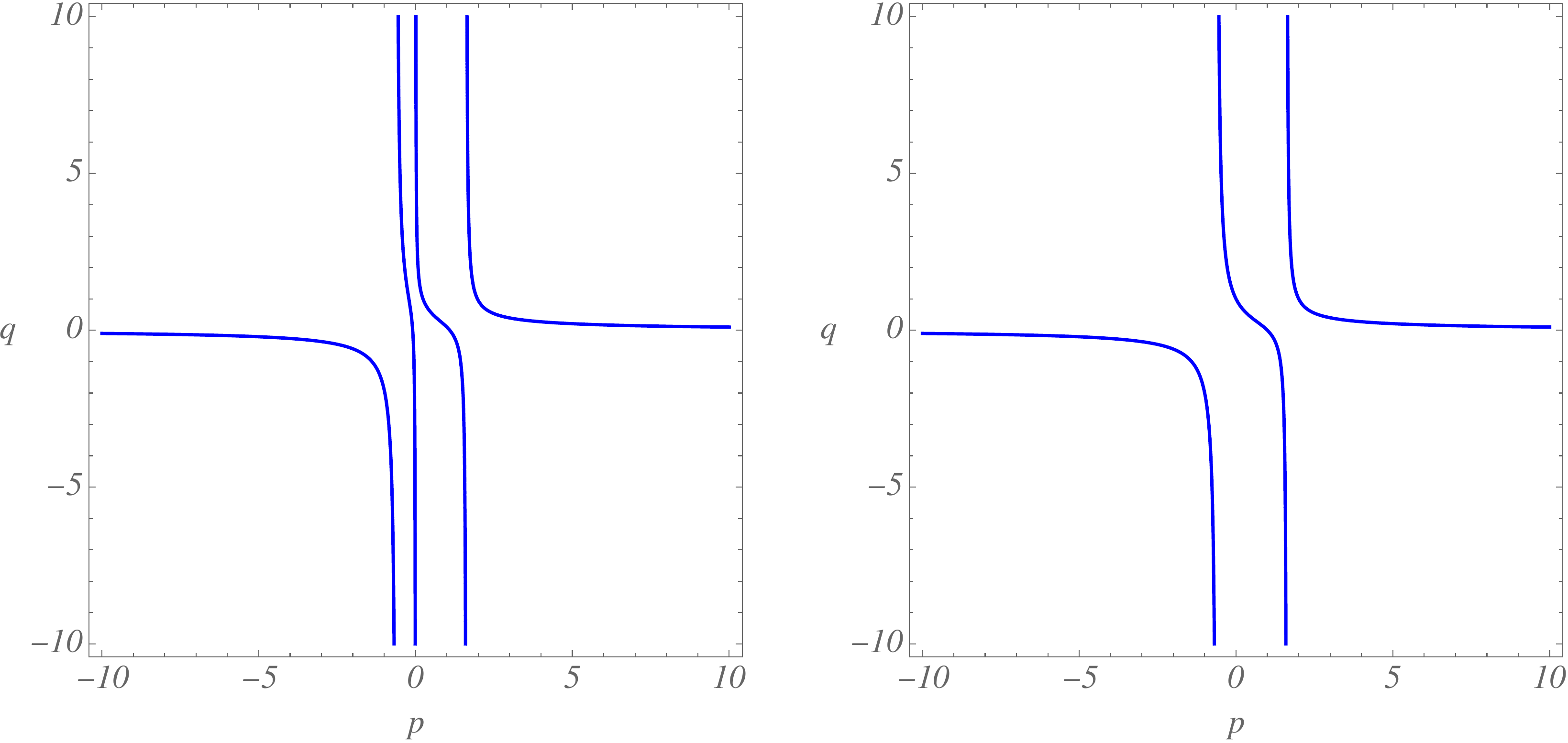}
\\
$(a)$ \hspace{4cm} $(b)$
\end{center}
\caption{$(a)$ Plot on the $(p, q)$-plane of the long wave limit curve defined by Eq.~\eqref{lwlimit_curve} for $\rho=0.9$, $H=1$, $\widetilde{\Omega}_1=1$. Panel $(b)$ illustrates the corresponding plot obtained for a homogeneous fluid ($\rho \rightarrow 1$).
%, with the same parameters of $H$ and $\widetilde{\Omega}_1$. 
\label{plot_lwlimit}
}
\end{figure}

%\begin{figure}[htbp]
%\begin{center}
%\includegraphics*[width=160pt]{figures/lw_limit_curve}
%\qquad 
%\includegraphics*[width=160pt]{figures/lw_limit_curve_hom_case}
%\\
%$(a)$ \hspace{4cm} $(b)$
%\end{center}
%\caption{$(a)$ Plot on the $(p, q)$-plane of the long wave limit curve defined by Eq.~\eqref{lwlimit_curve} for $\rho=0.9$, $H=1$, $\widetilde{\Omega}_1=1$. Panel $(b)$ illustrates the corresponding plot obtained for a homogeneous fluid ($\rho \rightarrow 1$).
%%, with the same parameters of $H$ and $\widetilde{\Omega}_1$. 
%\label{plot_lwlimit}
%}
%\end{figure}

In the long wave limit when $\alpha \rightarrow 0$, the curve \eqref{curve} reduces to 
\begin{equation}\label{lwlimit_curve}
(p^2-\widetilde{\Omega}_1 p-1)p q-H(p^2-\widetilde{\Omega}_1 p + \rho-1)=0,
\end{equation}
which can be parameterized as 
$$
q= \frac{H (p^2-\widetilde{\Omega}_1 p -1 +\rho)}{p (p^2-\widetilde{\Omega}_1 p-1)}.
$$
The curve is thus simply the graph of a real-valued function. It is straightforward to show the existence of three distinct vertical asymptotes with equations $p=p_i$ ($i=1,2,3$), and one horizontal asymptote $q=0$. Moreover, at each one of the points $p_i$, we have $\lim_{p\rightarrow p_i^{\pm}} q (p) = \pm \infty$. As a result, any line with slope one, as in \eqref{relationship_p_q}, intersects the curve at four distinct points (see Fig.~\ref{plot_lwlimit}$(a)$), and stability holds.
%. Hence, in the long wave limit, stability holds regardless of the physical parameters considered. 

The same analysis conducted for a homogeneous fluid yields a curve that can be parameterized as
$q= H (p-\widetilde{\Omega}_1) / (p^2-\widetilde{\Omega}_1 p-1)$, for which stability clearly holds (see Fig.~\ref{plot_lwlimit}$(b)$), in agreement with the findings of Ref.~\cite{Chesnokov_et_al}. 
A stratification in density has therefore no effect on the long wave stability of the flow.

\subsubsection{Case of a constant background current in the upper layer ($\Omega_2=0$)}

\begin{figure}[htbp]
\begin{center}
\includegraphics*[width=340pt]{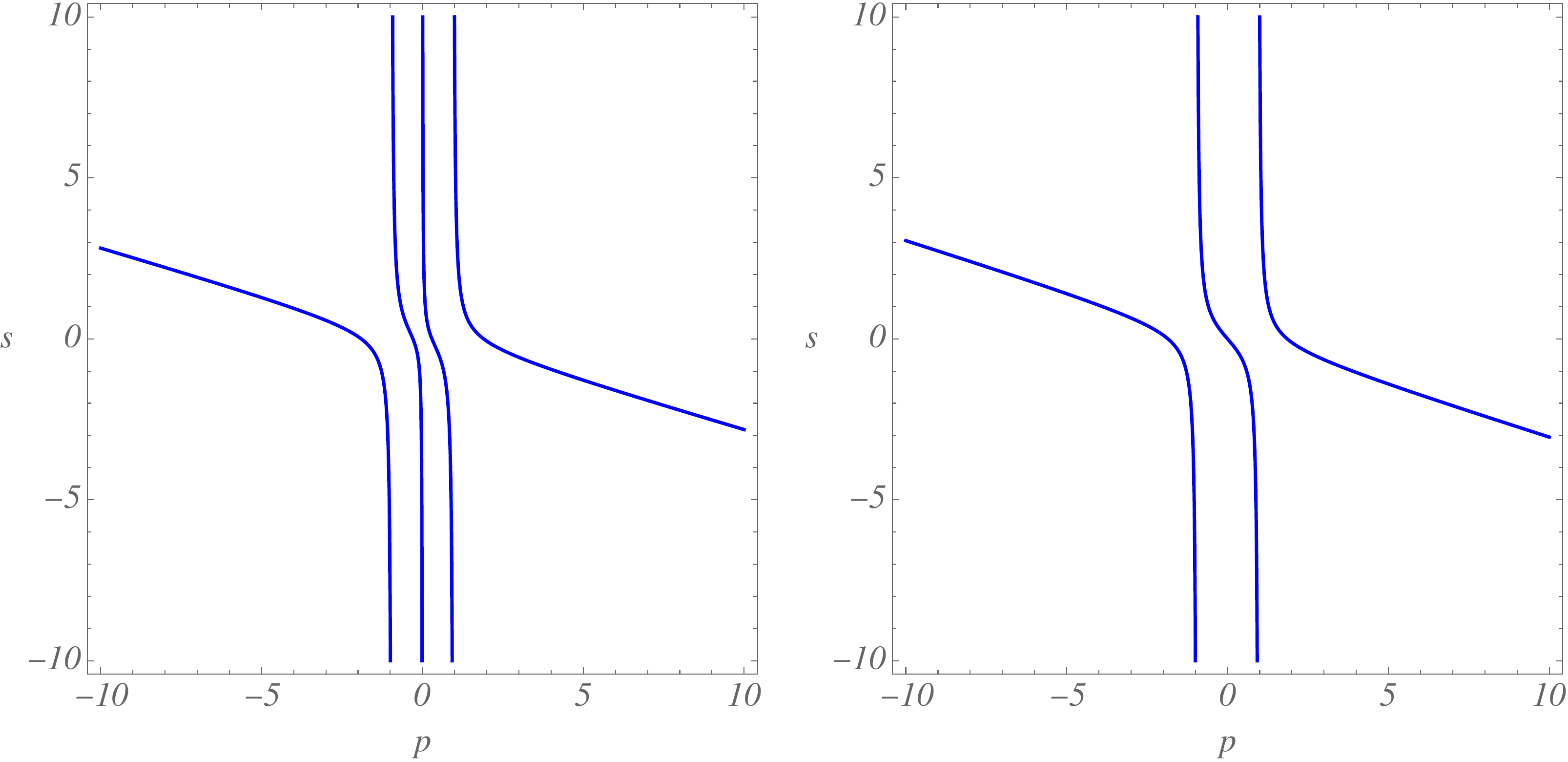}
\\
$(a)$ \hspace{4cm} $(b)$
\end{center}
\caption{$(a)$ Plot on the $(p, s)$-plane of the curve defined by Eq.~\eqref{curve_irrot_upper_layer} for $\rho=0.9$, $H=1$, $\alpha=0.5$. Panel $(b)$ illustrates the corresponding plot obtained for a homogeneous fluid ($\rho \rightarrow 1$).
\label{plot_irrot_upper_layer}
}
\end{figure}

%\begin{figure}[htbp]
%\begin{center}
%\includegraphics*[width=160pt]{figures/curve_irrot_upper_layer}
%\qquad 
%\includegraphics*[width=160pt]{figures/curve_irrot_upper_layer_hom_case}
%%\includegraphics*[width=119pt]{figures/curve_alpha=1}
%%\quad 
%%\includegraphics*[width=150pt]{figures/curve_alpha_large}
%\\
%$(a)$ \hspace{4cm} $(b)$
%\end{center}
%\caption{$(a)$ Plot on the $(p, s)$-plane of the curve defined by Eq.~\eqref{curve_irrot_upper_layer} for $\rho=0.9$, $H=1$, $\alpha=0.5$. Panel $(b)$ illustrates the corresponding plot obtained for a homogeneous fluid ($\rho \rightarrow 1$).
%\label{plot_irrot_upper_layer}
%}
%\end{figure}

As previously mentioned, the vorticity in the upper layer does not appear as a parameter in the curve equation \eqref{curve}. For this reason, to examine the case when the vorticity in the upper layer vanishes, we need to go back to the original eigenvalue equation \eqref{dispersion_relation}. If the condition $\Omega_2=0$ is inserted directly into \eqref{dispersion_relation}, one obtains a polynomial equation of degree $4$ in terms of the single variable $(u_0-c)$, since the background velocity $\hat{u}$ at the undisturbed level $z=H_2$ is identical to the velocity $u_0$ at $z=0$.
We next introduce new variables to take advantage of the geometrical formulation.
%In this case, the background velocity $\hat{u}$ at the undisturbed level $z=H_2$ is identical to the velocity $u_0$ at $z=0$. Hence, replacing $\Omega_2=0$ directly into Eq.~\eqref{dispersion_relation} yields a polynomial equation of degree $4$ in terms of the single variable $(u_0-c)$. 
%To take advantage of the geometrical formulation, a new variable needs to be introduced.
Let $\check{u}$ denote the background velocity at the bottom $z=-H_1$. Then, we have 
$$
\Omega_1 = [(\hat{u}-c)- (\check{u}-c)]/H_1,
$$
thus a natural step to take would be introducing a new variable $s$ defined by 
$$
s=(\check{u}-c)/{\sqrt{gH_1}},
$$
and casting the eigenvalue equation into the form
\begin{equation*}
\big[ \alpha\,\coth \alpha \,p^2 - (p-s) p -1 + \rho \left( \alpha \,\coth(H\alpha) \,p^2+1 \right) \big]
\left[ \alpha \,\coth(H\alpha) \,p^2 -1 \right] =\rho \,\alpha^2 \,\csch^2(H \alpha) \,p^4,
\end{equation*}
%Here, we have used the following relationship 
%$$
%\Omega_1 (u_0-c)/g  = \frac{H_1 \Omega_1}{\sqrt{gH_1}} \frac{(u_0-c)}{\sqrt{gH_1}} = (p-s) p, 
%$$
or equivalently:
\begin{equation}\label{curve_irrot_upper_layer}
\big[ \beta_1 \,p^2 - (p-s) p -1 \big] \left( \beta_2 \, p^2 -H \right) = \rho H \left(1-\alpha^2 \,p^4 \right),
\end{equation}
by using the identity $\csch^2 (H\alpha) = \coth^2 (H\alpha)-1$. In turn, this curve can be parameterized as
$$
s=-\frac{A}{\beta_2} \,p + \frac{B\,p^2 + \beta_2 H (1-\rho)}{\beta_2 \,p \,(H-\beta_2 p^2)},
$$
where $A$ and $B$ are given by
$$
A= \beta_2 (\beta_1-1)+ \rho H \alpha^2, \quad B = \rho H^2 \alpha^2 - \beta_2^2.  
$$

It readily follows the existence of three vertical asymptotes with equations $p=p_i$ ($i=1,2,3$) and one oblique asymptote with negative slope ($-A/\beta_2<0$). 
Moreover, for each one of the points $p_i$ it can be shown that $\lim_{p\rightarrow p_i^{\pm}} q (p) = \pm \infty$. As a consequence, any line with equation $s=p-\widetilde{\Omega}_1$ intersects the curve at four distinct points and stability holds (see Fig.~\ref{plot_irrot_upper_layer}$(a)$). Similarly as above, it can be verified that lack of stratification in density has no effect on the stability of the flow, as illustrated in Fig.~\ref{plot_irrot_upper_layer}$(b)$.  
%regardless of the physical parameters considered. 

In Ref.~\cite{Curtis_et_al}, it is claimed that the suppression of instability for a weak upper layer shear is a generic feature. While we can confirm that instability is suppressed when the vorticity in the upper layer vanishes, we will disprove their assertion by showing that in any other case the flow is unstable (see \S~\ref{sec:anal_gen_case}).

\subsection{Unstable configurations}

\subsubsection{Case of a constant background current in the lower layer ($\Omega_1=0$)}\label{sec:omega_1_is_0}

%This particular configuration has been well studied in the framework of a homogeneous fluid (see e.g. Miles 1957; Caponi {\it et al.} 1992). 
It is well known that when a homogeneous fluid current is modeled as a finite layer of constant vorticity above a finite stagnant region, and surface tension is neglected, the flow is unstable (see e.g. \cite{Miles,Caponi_et_al}). Here, we investigate whether stratification in density can stabilize the flow. 
\begin{figure}[htbp]
\begin{center}
\includegraphics*[width=340pt]{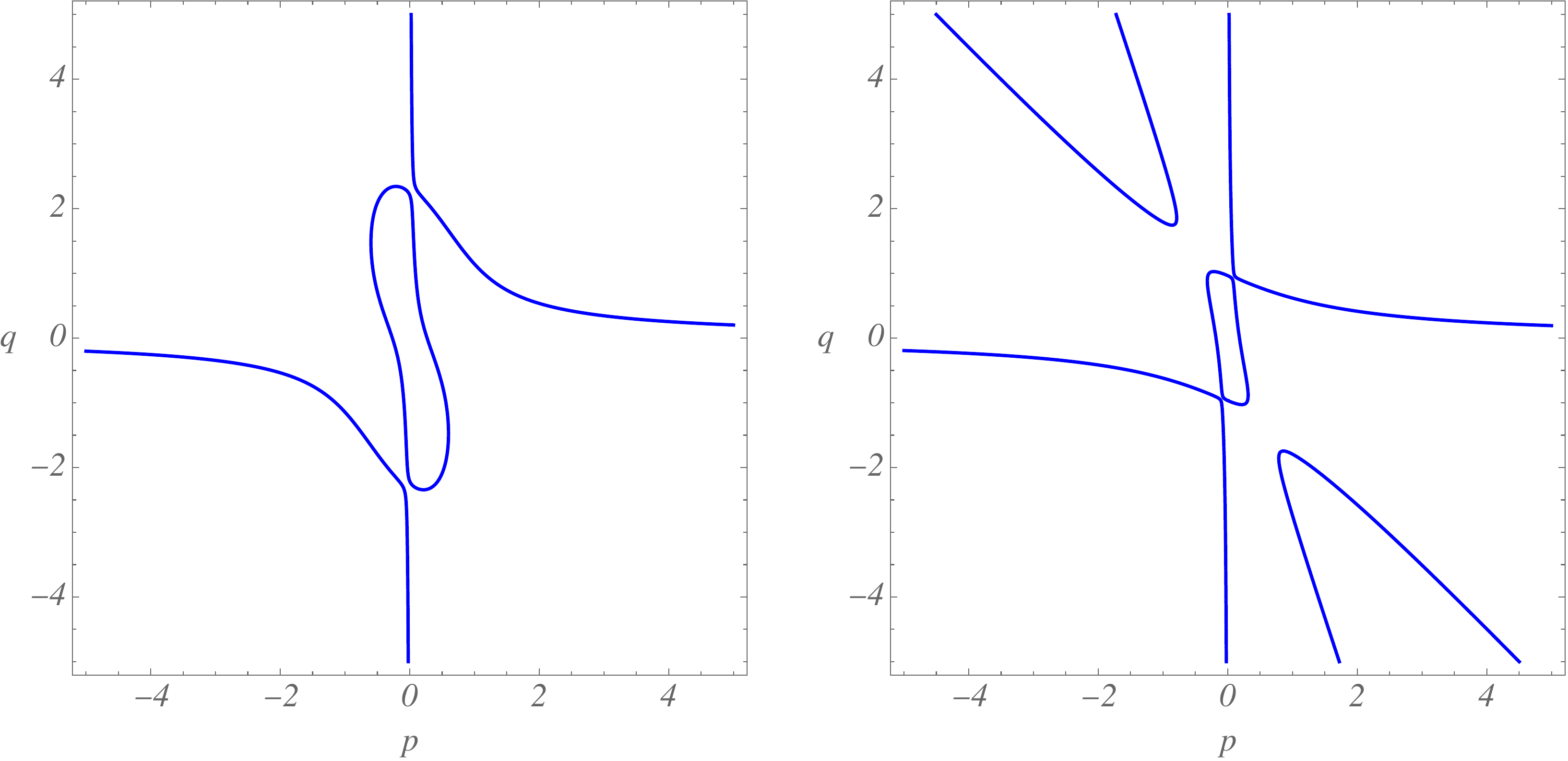}
\\
$(a)$ \hspace{4cm} $(b)$
\end{center}
\caption{Plots on the $(p, q)$-plane of the curve defined by Eq.~\eqref{curve} for $\rho=0.9$, $H=1$, $\widetilde{\Omega}_1=0$.
$(a)$ $\alpha =0.8$. $(b)$ $\alpha =2.0$. For these values of $\rho$, $H$, and $\widetilde{\Omega}_1$, when the value of $\alpha$ is gradually increased, a transition between the two possible curve configurations occurs at $\alpha \approx 1.77699$ (see Fig.~\ref{transition_configurations_irrotational_lower_layer}).
\label{plot_irrotational_lower_layer}
}
\end{figure}

%\begin{figure}[htbp]
%\begin{center}
%\includegraphics*[width=160pt]{figures/Omega1_is_0_small_alpha}
%\qquad 
%\includegraphics*[width=160pt]{figures/Omega1_is_0_large_alpha}
%\\
%$(a)$ \hspace{4cm} $(b)$
%\end{center}
%\caption{Plots on the $(p, q)$-plane of the curve defined by Eq.~\eqref{curve} for $\rho=0.9$, $H=1$, $\widetilde{\Omega}_1=0$.
%$(a)$ $\alpha =0.8$. $(b)$ $\alpha =2.0$. \tcb{For these values of $\rho$, $H$, and $\widetilde{\Omega}_1$, when the value of $\alpha$ is gradually increased, a transition between the two possible curve configurations occurs at $\alpha \approx 1.77699$ (see Fig.~\ref{transition_configurations_irrotational_lower_layer}).}
%\label{plot_irrotational_lower_layer}
%}
%\end{figure}

To determine the number of intersections between the curve \eqref{curve} and the line \eqref{relationship_p_q} as 
the Froude number increases, the behavior of an algebraic curve at infinity
%described here by the asymptotes
%or nonlinear branches 
%at infinity, 
is important as the $q$-intercept of the line has precisely the value $F$.

It is convenient to express the curve equation \eqref{curve} as
$P(p,q) \equiv \sum_{k=0}^{4} P_k(p,q)=0$, where $P_k (p,q)$ is a homogeneous polynomial in $p$ and $q$ of degree $k$, since the factors of the highest degree polynomial define the slopes of the asymptotes to the curve. When vorticity is absent in the lower layer ($\widetilde{\Omega}_1=0$), the curve \eqref{curve} satisfies $P(-p,-q)=P(p,q)$, {\it i.e,} the curve $P(p,q)=0$ is symmetric about the origin, and no odd degree polynomials show up in the curve equation. This allows us to use the following  results (Primrose~\cite{Primrose}, Theorem 2, pp.~7--8): $(i)$ {\it any simple factor $ap+bq$ of $P_4(p,q)$ will have associated an asymptote to the curve, defined by the equation $ap+bq =0$; $(ii)$ if $ap+bq$ is a repeated factor of $P_4(p,q)$, so that $P_4(p,q)=(ap+bq)^2 Q(p,q)$, then it will have associated at most two possible asymptotes $ap+bq=t_0$, where $t_0$ is a real root of $Q(b,-a) \,t^2 +P_2(b,-a)=0$}. 
In both cases, using homogeneous coordinates, $(b,-a,0)$ is a point at infinity. However, only in $(ii)$ the point $(b,-a,0)$ is a singular point. As a result, nonlinear branches at infinity do not exist and the behavior of the curve at infinity is hence  completely described by its asymptotes (see Appendix C in \cite{Barros_Choi_2014}).
%, as mentioned earlier. 

Here, the highest degree polynomial $P_4 (p, q)$ is given by
\begin{equation}\label{P4_exp}
P_4 (p,q)= pq \Big[ (H \beta_1 +\rho (\beta_2 -1)) p^2 + (-H\beta_1 + H \beta_1 \beta_2 + 2\rho- 2\rho \beta_2+ \rho \beta_2^2 - \rho \beta_3) \,pq + \rho (\beta_2-1) q^2 \Big], 
\end{equation}
from which it is found that, regardless of parameters used, the curve has two asymptotes: $p = 0$ and $q = 0$. 
%It can be further proven that the asymptote $p = 0$ for large negative values of $q$ can only be connected to the asymptote $q = 0$ for large negative values of $p$ in the third quadrant. By symmetry, such a real connected component also exists in the first quadrant. 
To examine the existence of additional asymptotes to the curve, we have to take into account the terms in the bracket in \eqref{P4_exp} and consider the following quadratic equation for $v=p/q$:
%$$
%(H \beta_1 +\rho (\beta_2 -1)) v^2 + (-H\beta_1 + H \beta_1 \beta_2 + 2\rho- 2\rho \beta_2+ \rho \beta_2^2 - \rho \beta_3) \,v + \rho (\beta_2-1)=0. 
%$$
\begin{equation}\label{quadratic_eq_for_v}
(H \beta_1 +\rho (\beta_2 -1)) v^2  + \left( H \beta_1 (\beta_2-1)+\rho(\beta_2 -1)^2 + \rho (1-\beta_3) \right)\,v + \rho (\beta_2-1)=0.
\end{equation}
The discriminant of this quadratic equation fully determines the number of real roots and, when positive, two extra asymptotes to the curve exist. If so, their slopes must be negative, since the coefficients in \eqref{quadratic_eq_for_v} are all positive.
%Moreover, since the coefficients in \eqref{quadratic_eq_for_v} are all positive, if two distinct real roots exist, both must be negative. 
Incidentally, for this example, it may noted that 
%it can be shown that 
when the discriminant vanishes, {\it i.e.}, $P_4$ has a repeated factor, there are no extra asymptotes to the curve.  

%We will describe its behavior for small and large values of $\alpha$. 
When $\alpha$ is small ($0<\alpha \ll 1$), we can consider the expansions: 
\begin{equation}\label{betas_expansions}
\beta_1 = 1 + \frac{1}{3} \alpha^2 + O(\alpha^4), \quad \beta_1 = 2 + \frac{1}{3} H^2 \alpha^2 + O(\alpha^4), \quad \beta_3 = 1 - \frac{1}{3} H^2 \alpha^2 + O(\alpha^4),
\end{equation}
under which the discriminant may be approximated by $(-4/3) \rho H^3 \alpha^2 + O(\alpha^4)$, thus negative.
%$$
%\alpha ^2 H^3 \Big(-12 \rho + (H^3-4\rho+ 2\rho H^2 - 3\rho^2 H) \alpha^2 + O(\alpha^4)\Big)/9,
%$$
%which is negative, provided $\alpha$ is small enough. 
On the other hand, for large values of $\alpha$ ($\alpha \gg 1$), {\it i.e.,} for short-length waves, the discriminant is well approximated by 
$H^2 \alpha ^2 (\rho H \alpha + H \alpha -2 \rho -1)^2$, thus positive. 
\begin{figure}[htbp]
\begin{center}
\includegraphics*[width=170pt]{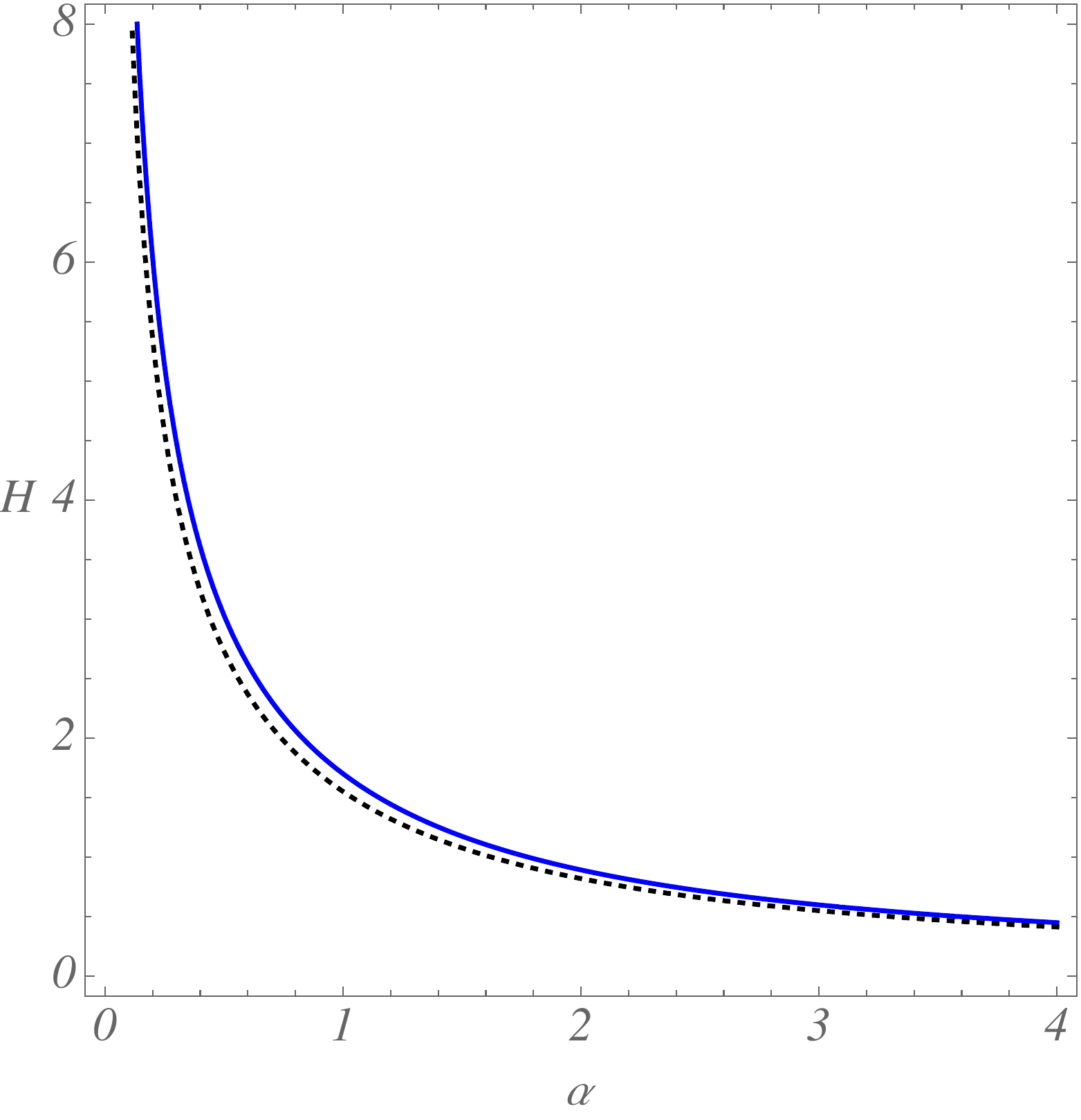}
\end{center}
\caption{Location in the parameter region where a transition between the two configurations for the family of curves defined by Eq.~\eqref{curve}, with $\widetilde{\Omega}_1=0$, occurs. The transition is predicted by the vanishing of the discriminant for the quadratic \eqref{quadratic_eq_for_v}. Different values of $\rho$ are considered: $\rho=0.9$ (full line); $\rho = 0.6$ (dotted line). 
%In panel $(b)$ we identify the location in the parameter space where the discriminant of the quadratic \eqref{quadratic_eq_for_v} vanishes. Here, $0<\alpha<4$, $0<H<8$, $0<\rho<1$.  
\label{transition_configurations_irrotational_lower_layer}
}
\end{figure}
%\begin{figure}[htbp]
%\begin{center}
%\includegraphics*[width=170pt]{figures/trans_config_omega1_is_0}
%\end{center}
%\caption{\tcb{Location in the parameter region where a transition between the two configurations for the family of curves defined by Eq.~\eqref{curve}, with $\widetilde{\Omega}_1=0$, occurs. The transition is} predicted by the vanishing of the discriminant for the quadratic \eqref{quadratic_eq_for_v}. Different values of $\rho$ are considered: \tcb{$\rho=0.9$ (full line); $\rho = 0.6$ (dotted line)}. 
%%In panel $(b)$ we identify the location in the parameter space where the discriminant of the quadratic \eqref{quadratic_eq_for_v} vanishes. Here, $0<\alpha<4$, $0<H<8$, $0<\rho<1$.  
%\label{transition_configurations_irrotational_lower_layer}
%}
%\end{figure}
This suggests the existence of, at least, two different configurations for our plane algebraic curve, as shown in Fig.~\ref{plot_irrotational_lower_layer}. Notice that not only the behavior at infinity has changed, but also the topological structure of the curve. A classical result in Algebraic Geometry asserts that {\it as we vary the parameters, the topological structure of the curve changes only when the coefficients pass through values for which the curve has a singularity} (see Lemma 1 in \cite{Petrowsky}). 
Singular points at infinity exist whenever $P_4$ has a repeated factor, {\it i.e.} when the discriminant of \eqref{quadratic_eq_for_v} vanishes. 
We now show that  {\it finite singular points} (solutions of $P=0$, $P_p=0$, and $P_q=0$) cannot exist for the particular curves under consideration. 
%To prove that these are the only possibilities, it remains to 
%We show the non-existence of {\it finite singular points} (solutions of $P=0$, $P_p=0$, and $P_q=0$). 
Following Barros and Choi~\cite{Barros_Choi_2014}, we wish to prove that the system  
\begin{equation}\label{cand_singular_pts}
2 P_0 P_{4,p} - P_2 P_{2,p} =0, \quad 2 P_0 P_{4,q} - P_2 P_{2,q} =0, 
\end{equation}
has no solutions. Here, the subscripts $p$ and $q$ denote partial differentiation with respect to $p$ and $q$, respectively. Given that \eqref{cand_singular_pts} constitutes a system of two homogeneous polynomials of degree 3 in the variables p and q,  we may introduce $v=p/q$ to write a system of two cubic polynomials in $v$. The cubics have a common root if the {\it resultant} vanishes \cite{Prasolov}. Thus, finite singular points exist provided
$$
64 (\beta_2-1)^2 \beta_3^2 (\rho -1)^4 \rho ^2 \left(\rho (\beta_2-1) +H \beta_1\right)^2 (a_0 \rho^4 + a_1 \rho^3 + a_2 \rho^2+a_3 \rho +a_4) =0,
$$
with coefficients $a_i$ ($i=0,1,2,3,4$) depending on the parameters $H$ and $\alpha$. 
%too cumbersome to be presented here. 
The quartic in $\rho$ is, however, identical to the term $e_4$ in the expression of the function $G$ defined in \eqref{quartic_for_lambda}, which can never vanish (see Appendix D), thus the result follows.

%the expression above must vanish. Our extensive numerical tests confirm, however, that this can never happen, 
%We can then conclude the non-existence of finite singular points. 

It is important to emphasize that there could be, in principle, other curve configurations leading to different stability properties of the flow, whilst preserving the same topological structure (see Example in Appendix E). This circumstance can, however, be ruled out by examining the possible arrangements of the curve branches connecting the asymptotes (see further details in Appendix A).
%However, as shown in Appendix A, there is one single possible arrangement of the curve branches connecting the asymptotes.
As a consequence,  
%This implies that, 
when $\widetilde{\Omega}_1=0$, the family of curves defined by \eqref{curve} has precisely two distinct configurations. For fixed values of $\rho$ and $H$, the transition between the two configurations occurs only once (see Fig.~\ref{transition_configurations_irrotational_lower_layer}), and at the expense of a singular point at infinity. 
In the example depicted in Fig.~\ref{plot_irrotational_lower_layer}, with $\rho=0.9$, $H=1$, when the value of $\alpha$ is gradually increased, a transition between a curve with two asymptotes and a curve with four asymptotes will take place at $\alpha\approx 1.77699$.

Every curve in the family cuts the axes at four points close to the origin, belonging to the contour represented in Fig.~\ref{plot_irrotational_lower_layer} by an oval, which (for fixed values of $\rho$, $H$, and $\alpha$) guarantees the stability of the flow for small values of $F$ (or $\widetilde{\Omega}_2$). In addition, the presence of two extra asymptotes as in Fig.~\ref{plot_irrotational_lower_layer}(b) assures the stability for large values of $F$. In either case, instability occurs at least for intermediate values of Froude numbers.
\begin{figure}[htbp]
\begin{center}
\includegraphics*[width=160pt]{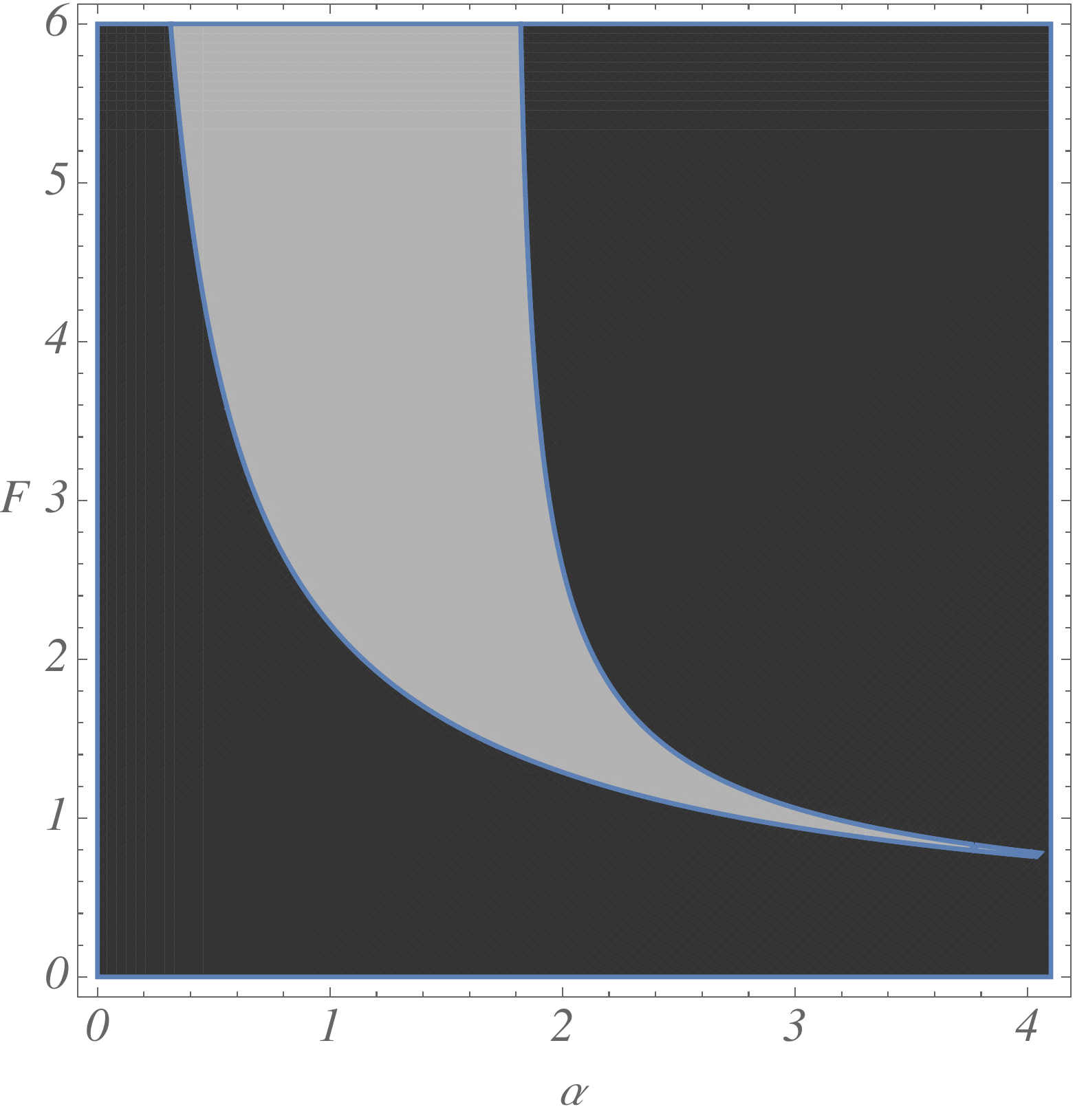}
\end{center}
\caption{Stability diagram on the $(\alpha, F)$-plane for the case of a constant background current in the lower layer, prescribed by Eq.~\eqref{curve} with $\rho=0.9$, $H=1$, and $\widetilde{\Omega}_1=0$. The dark shaded region corresponds to a stable region with four real roots and the light shaded region corresponds to the unstable region with two complex and two real roots. Hereafter, the same color scheme will be adopted. Only positive values of $F$ are considered, since symmetry of the curve about the origin implies that the stability diagram is symmetric with respect to the $\alpha$-axis. This same property holds for each case considered in this paper.
\label{stability_diagram_irrotational_lower_layer}
}
\end{figure}

%\begin{figure}[htbp]
%\begin{center}
%\includegraphics*[width=160pt]{figures/sdiagram_omega_1_is0}
%\end{center}
%\caption{Stability diagram on the $(\alpha, F)$-plane for the case of \tcb{a constant background current in the} lower layer, prescribed by Eq.~\eqref{curve} with $\rho=0.9$, $H=1$, and $\widetilde{\Omega}_1=0$. The dark shaded region corresponds to a stable region with four real roots and the light shaded region corresponds to the unstable region with two complex and two real roots. Hereafter, the same color scheme will be adopted. Only positive values of $F$ are considered, since symmetry of the curve about the origin implies that the stability diagram is symmetric with respect to the $\alpha$-axis. This same property holds for each case considered in this paper.
%\label{stability_diagram_irrotational_lower_layer}
%}
%\end{figure}

The geometrical approach presented here is, of course, equivalent to the standard analytical approach where a quartic equation for the phase velocity $c$, or equivalently for $p$, is obtained by substituting \eqref{relationship_p_q} into \eqref{curve}. Then, using Fuller's root location criteria \cite{Fuller,Jury_Mansour}, stability diagrams on the $(\alpha,F)$-plane can be drawn. 
%Typically, three colors would be required to distinguish a stable region from the two kinds of unstable regions (characterized by two, or four complex roots for the eigenvalue equation). In all diagrams presented here, however, two colors suffice, the reason being that four complex solutions to the eigenvalue equation \eqref{dispersion_relation} can never exist (see details in Appendix).
As shown in Fig.~\ref{stability_diagram_irrotational_lower_layer}, for any fixed non-zero value of $F$, there is always a finite range of wavenumbers for which a single unstable wave appears. Furthermore, the band of wavenumber for instability becomes narrower as $F$ decreases. The instability of the corresponding homogeneous fluid cannot therefore be suppressed by stratification in density, as it is insufficient to stabilise short-length waves ($\alpha\gg1$), even at very small (non-zero) values of $F$.

It is important to emphasize that, although this two-layer physical setting can formally be obtained from the three-layer configuration known as the Taylor-Goldstein configuration (see \cite{Taylor, Goldstein}), in the limit when the density of the top layer tends to zero, the stability results obtained here differ substantially from those for the three-layer flow where all densities are strictly positive. In particular, this two-layer shear current cannot be stabilized for sufficiently large difference in the streaming velocities $\hat{u}$ and $u_0$, unlike the Taylor-Goldstein configuration (see \S 3 in \cite{Barros_Choi_2014}).

\subsubsection{Case of uniform vorticity ($\Omega_1=\Omega_2$)}\label{sec:uniform_vorticity}

It is elementary to show that a homogeneous linear shear current is stable. The dispersion relation is found as a quadratic equation for the wave speed $c$, whose roots are always real. 
%This is another configuration that has been well studied in the framework of a homogeneous fluid. 
%A homogeneous linear shear current is known to be stable and several efforts have been made in order to investigate the properties of finite amplitude surface waves in running water. 
%since the seminal work by Teles da Silva \& Peregrine (1988). 
We will see here that stratification has a strong destabilizing effect on the flow, in the sense that any two-layer stratified flow with non-zero uniform vorticity is unstable.    
%In fact, we shall prove that the stratified shear flow defined in \eqref{physical_configuration} with $\rho_2<\rho_1$ is unstable for any non-zero uniform vorticity ($\Omega_1= \Omega_2 \neq 0$). 
%by given any stable stratification ($0<\rho_2/\rho_1<1$), the flow is unstable if $\Omega_1= \Omega_2 \neq 0$. 

Since the vorticity in the upper layer is not a parameter for the family of curves defined by \eqref{curve}, in order to address this particular case, we are required to go back to \eqref{dispersion_relation}. We are then faced with two alternatives. We could simply substitute $\Omega_2$ by $\Omega_1$ in \eqref{dispersion_relation} and proceed as in \S~\ref{sec:formulation} to obtain a similar plane algebraic curve to the one given in \eqref{curve}. Such curve, however, has no symmetry properties that allow us to easily study its singularities. An alternative to this would be using the fact that $\widetilde{\Omega}_1= \widetilde{\Omega}_2$ to insert \eqref{Omega2_of_p_q} directly into \eqref{curve}, yielding the following new family of curves:
\begin{equation}\label{curve_constant_vorticity}
\left[ H \beta_1 \, p^2 - p(q-p) -H + \rho \left( \beta_2 \,p^2+ p(q-p)+H \right) \right]  \left( \beta_2 \,q^2 - q(q-p)-H \right)  = \rho  \beta_3 \, p^2 q^2.
\end{equation}
\begin{figure}[htbp]
\begin{center}
\includegraphics*[width=340pt]{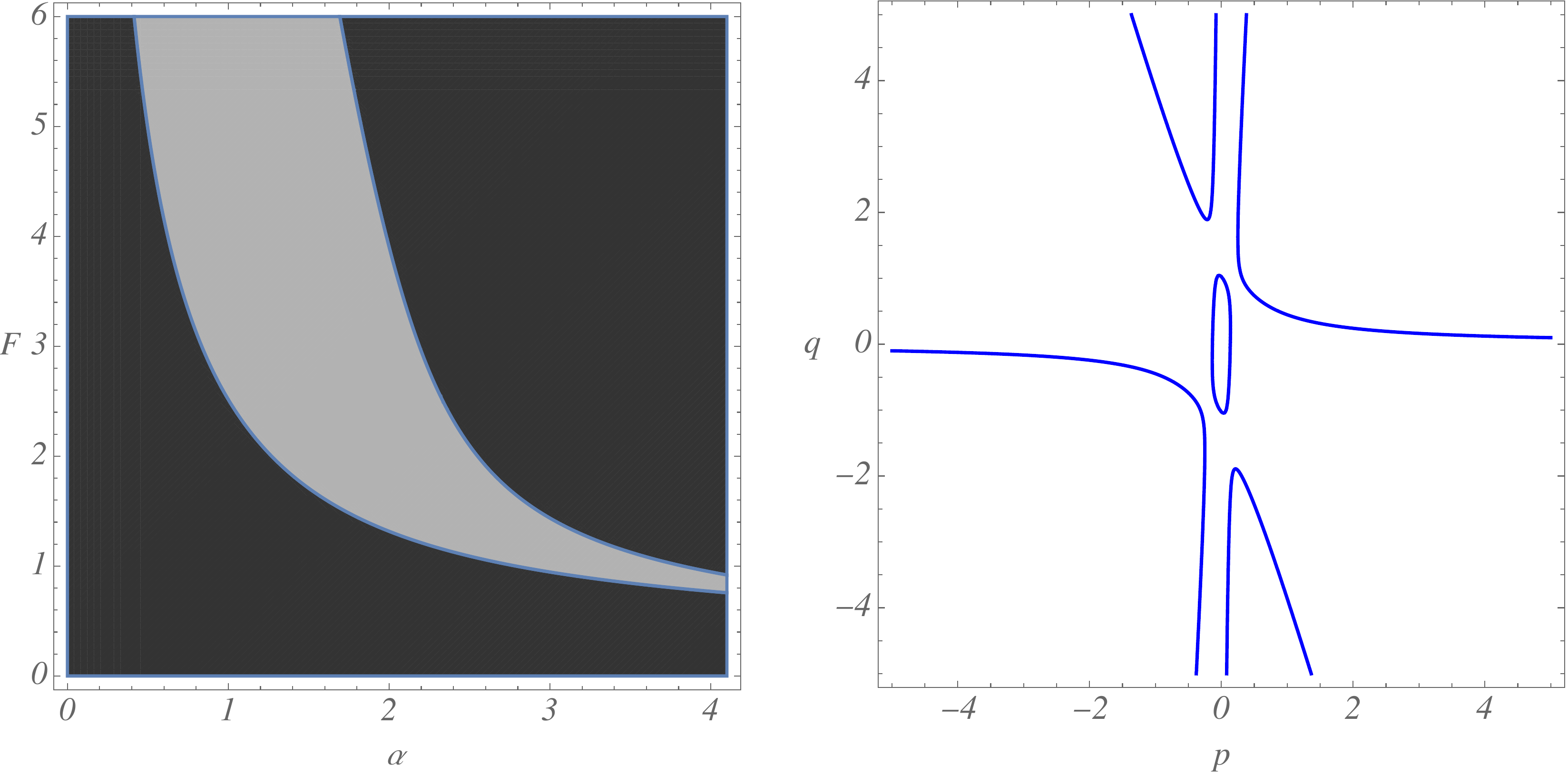}
\\
$(a)$ \hspace{4cm} $(b)$
\end{center}
\caption{$(a)$ Stability diagram on the $(\alpha, F)$-plane for the case of uniform vorticity given by Eq.~\eqref{curve_constant_vorticity} with $\rho=0.9$, $H=0.5$. Dark shaded: stable region; light shaded: unstable region. In panel $(b)$ we show the plot on the $(p, q)$-plane of the curve defined by Eq.~\eqref{curve_constant_vorticity} for $\rho=0.9$, $H=0.5$, and $\alpha =2.5$. 
\label{plot_diagram_constant_vorticity}
}
\end{figure}

%\begin{figure}[htbp]
%\begin{center}
%\includegraphics*[width=160pt]{figures/sdiagram_constant_vorticity}
%\qquad 
%\includegraphics*[width=160pt]{figures/constant_vorticity_large_alpha}
%\\
%$(a)$ \hspace{4cm} $(b)$
%\end{center}
%\caption{$(a)$ Stability diagram on the $(\alpha, F)$-plane for the case of uniform vorticity given by Eq.~\eqref{curve_constant_vorticity} with $\rho=0.9$, $H=0.5$. Dark shaded: stable region; light shaded: unstable region. In panel $(b)$ we show the plot on the $(p, q)$-plane of the curve defined by Eq.~\eqref{curve_constant_vorticity} for $\rho=0.9$, $H=0.5$, and $\alpha =2.5$. 
%\label{plot_diagram_constant_vorticity}
%}
%\end{figure}

Notice that all curves in this family are symmetric about the origin, which plays a key role in our analysis. Once again, by examining the highest degree polynomial, we find that $p = 0$ and $q = 0$ are asymptotes to any curve with Eq.~\eqref{curve_constant_vorticity}. The existence of additional asymptotes depend on the number of real roots of the quadratic equation 
%for $v=p/q$ below:
%The top degree form $P_4 (p, q)$ is given here by:
%\begin{multline*}\label{P4_new_exp}
%P_4 (p,q) = pq \Big[ (1+H \beta_1 + \rho (\beta_2-1)) p^2 + \\
%+(-2 + 2\rho  -H \beta_1 + \beta_2 + H \beta_1 \beta_2 -2 \rho \beta_2 + \rho \beta_2^2 - \rho \beta_3 ) \,pq +  (\beta_2-1) (\rho-1) q^2 \Big], 
%\end{multline*}
%from which it can be seen that the curve has the asymptotes $p=0$, and $q=0$. To inspect the existence of extra asymptotes, we consider the following quadratic equation for $v=p/q$:
%we need to establish the existence of two distinct roots for the quadratic equation for $v=p/q$:
%\begin{equation}\label{quad_uniform_vorticity}
% (1 +H \beta_1 + \rho (\beta_2-1)) v^2 + (2(\rho-1) -H \beta_1 + \beta_2 + H \beta_1 \beta_2 -2 \rho \beta_2 + \rho \beta_2^2 - \rho \beta_3 ) \,v +  (\beta_2-1) (\rho-1) =0.
%\end{equation}
\begin{equation}\label{quad_uniform_vorticity}
 (1 +H \beta_1 + \rho (\beta_2-1)) v^2 +  (\beta_2-2+ H\beta_1 (\beta_2-1) + \rho (\beta_2-1)^2 + \rho (1-\beta_3) ) \,v +  (\beta_2-1) (\rho-1) =0.
\end{equation}
The discriminant of this quadratic is given by $b_0 \rho^2 + b_1 \rho + b_2$, with coefficients:
\begin{align*}
b_0 &=H^2 \alpha ^2 \left(4 -4 \beta_2 +H^2 \alpha^2 \right),\\
b_1&= 2 H \left(-2 \beta_1 \beta_2^2+ 2 \beta_1 \beta_2+ H^2 \alpha ^2 
   \beta_1\beta_2 - H^2 \alpha^2 \beta_1 + H \alpha^2 \beta_2 - 2 H \alpha ^2 \right),\\
b_2&=(\beta_2+  H \beta_1 (\beta_2 -1))^2.
\end{align*}
Here, we have used the relationship $\beta_3 = \beta_2^2 - H^2 \alpha^2$. It is easy to see that the parabola $b_0 \rho^2 + b_1 \rho + b_2$ would have two real roots, if $\rho$ was allowed to take values on the whole real line. However, since $\rho$ is confined between $0$ and $1$ and the conditions in Corollary~\ref{no_roots_interval} (see Appendix C) are fulfilled, we conclude that $b_0 \rho^2 + b_1 \rho + b_2>0$. On that account, two more asymptotes to the curve \eqref{curve_constant_vorticity} must exist, with equations $p=v_0^{\pm} q$, where $v_0^{\pm}$ are real roots (with opposite signs) of \eqref{quad_uniform_vorticity}.
%Here, the highest-order coefficient is positive, while the constant term is negative, implying that the roots $v_0^{\pm}$ have opposite signs. 
We can then guarantee (for fixed values of $\rho$, $H$, and $\alpha$) the stability of the flow for large values of $F$ (see Fig.~\ref{plot_diagram_constant_vorticity}$(b)$).

%provided $v_0=1$ is not a root of this quadratic. Indeed, when $v$ is replaced by $1$, we obtain 
%$$
%H \beta_1 \beta_2 + \rho (\beta_2^2-\beta_3)
%$$
%which reduces to 
%$$
%H \beta_1 \beta_2 + \rho (H \alpha)^2,
%$$
%given that $\beta_3 = \beta_2^2 - (H \alpha)^2$. 
The fact that the top degree form $P_4$ does not have repeated factors implies the non-existence of singular points at infinity. To examine the existence of finite singular points, 
%With that being said, the family of curves given by \eqref{curve_constant_vorticity} must have one single configuration, provided it has no finite singular points.   
%We confirm below that this is the case. 
we consider the system of equations \eqref{cand_singular_pts}, by introducing the variable $v=p/q$, and seek conditions under which the two cubic polynomials in $v$ have a common root.
%
%We will establish that the system of equations \eqref{cand_singular_pts} has no solutions. Using the same strategy as before, we introduce the variable $v=p/q$ to examine the conditions under which the two cubic polynomials in $v$ have a common root. 
This amounts to imposing the vanishing of the resultant, given here by the expression:
$$
16 (\beta_2-1)^2 \beta_3^2 (\rho-1)^3 \rho^2 (1+ H \beta_1 + \rho (\beta_2-1))^2 (b_0  \rho^2 + b_1\rho + b_2).
$$
As we know from above, this expression is simply the product of positive terms, leading to the
%Hence the 
%system of equations \eqref{cand_singular_pts} has no solutions, implying the 
non-existence of finite singular points to the curve. Similarly to the previous case, findings in Appendix E, describing the way the branches of the curve approach the asymptotes, guarantee that one single configuration exists for the family of curves \eqref{curve_constant_vorticity}. As a result, for any fixed wavenumber $\alpha$ there exists a finite range of Froude numbers $F$ at which the flow is unstable. Moreover, as depicted in Fig.~\ref{plot_diagram_constant_vorticity}$(a)$, for any fixed value of $F$, there is a limited band of wavenumber for instability, which becomes narrower as $F$ decreases. 

\begin{figure}[htbp]
\begin{center}
\includegraphics*[width=430pt]{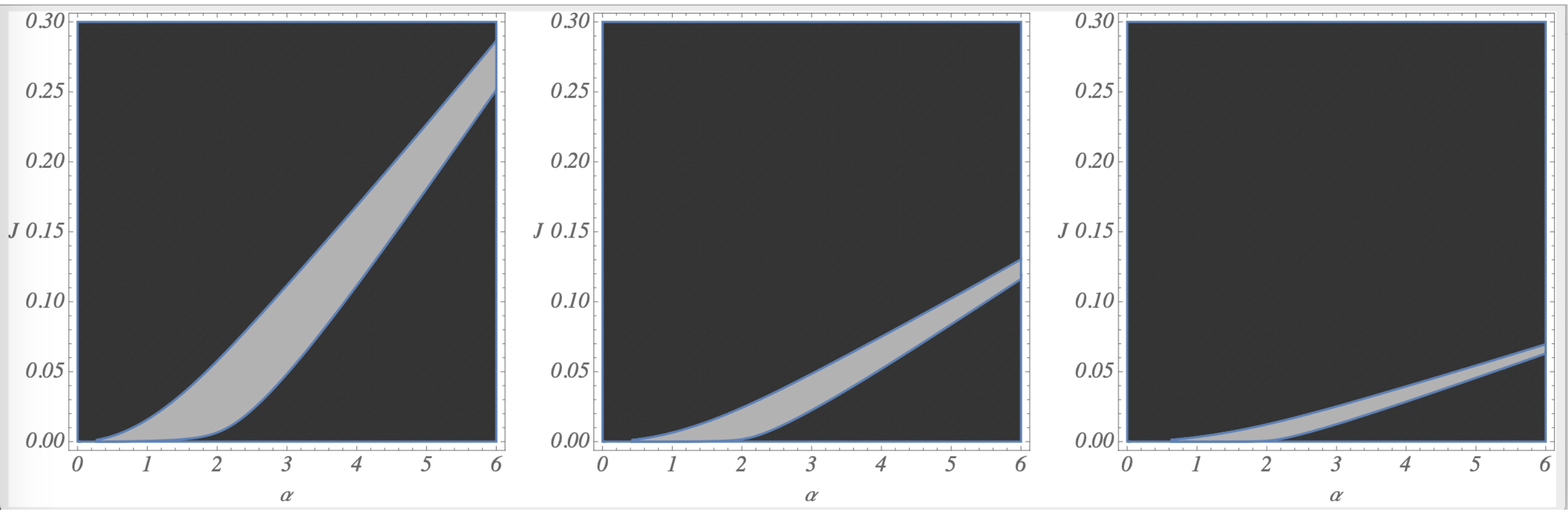}
\\
$(a)$ \hspace{4cm} $(b)$ \hspace{4cm} $(c)$
\end{center}
\caption{Stability diagrams on the $(\alpha, J)$-plane for the case of uniform vorticity given by Eq.~\eqref{curve_constant_vorticity} with $H=0.5$. Different values of $\rho$ are considered. $(a)$ $\rho=0.9$. $(b)$ $\rho=0.96$. $(c)$ $\rho=0.98$. Dark shaded: stable region; light shaded: unstable region. In all diagrams, the stability boundaries form a cusp at the origin.  
\label{stability_diagram_constant_vorticity_Richardson}
}
\end{figure}

In stratified flows, stability diagrams are often presented in terms of a Richardson number. Here, we can define the Richardson number $J$ by
$$
J= \left(\frac{\rho_1-\rho_2}{\rho_1}\right) \frac{gH_1}{(\hat{u}-u_0)^2},
$$
and visualize the stability diagram on the $(\alpha,J)$-plane, as in Fig.~\ref{stability_diagram_constant_vorticity_Richardson}. 
%The stability boundaries form a cusp at the origin. 
As the value of the density ratio $\rho$ increases, approaching 1, the less evident the instability region becomes. 
%the narrower becomes the instability region. 
This is expected, since the homogeneous case, known to be stable, can be recovered in the limit when $\rho\rightarrow 1$.

Similarly to the previous case, this two-layer physical setting can be seen as a special case of the three-layer configuration known as the Taylor configuration, when the density in the upper layer tends to zero. As long as all densities are strictly positive, it can be shown that Taylor's configuration is stable at low Richardson number (see \S 4 in \cite{Barros_Choi_2014}). Clearly this cannot be the case here, since we have proved that for any fixed wavenumber $\alpha$ there exists a finite range of values of the Richardson number $J$ at which the two-layer uniform shear current is unstable.

We would also like to point out that the infinite-depth case cannot be directly recovered from our analysis, because of the non-dimensional scaling used here. However, the instability of the flow still holds and can be physically explained based on the wave-interaction mechanism, as proposed by Carpenter, Tedford, Heifetz, and Lawrence (see \S 5.1 in \cite{Carpenter_et_al}).
%based on the wave interaction mechanism theory.}

%It is curious that the same quadratic form for $\rho$ shows up once again, which, as we know, can never vanish. The validity of our claim follows. 

%\begin{figure}[htbp]
%\begin{center}
%\includegraphics*[width=150pt]{figures/sdiagram_constant_vorticity}
%\end{center}
%\caption{Stability diagrams on the $(\alpha, F)$-plane for the constant vorticity case with $\rho=0.9$, $H=1$. We are considering here only positive values of $F$, since symmetry about the origin implies that the stability diagram for negative values of $F$ ({\it i.e.,} negative values for the vorticity) can be obtained by mirror symmetry of this digram with respect to the $\alpha$-axis. Dark shaded: stable region; light shaded: unstable region.
%\label{stability_diagram_constant_vorticity}
%}
%\end{figure}

\section{Stability analysis for the general case}\label{sec:anal_gen_case}

As stressed earlier, the analysis of the curve \eqref{curve} is in general rather difficult, when it comes to examining its singularities to fully classify the configurations obtained for this family of curves. Based on the idea proposed in \S~\ref{sec:uniform_vorticity}, we introduce a new parameter $\lambda$, such that $\Omega_1 = \lambda \Omega_2$, to cast \eqref{curve} into the form
\begin{equation}\label{curve_gen_case}
\left[ H \beta_1 \, p^2 - \lambda p(q-p) -H + \rho \left( \beta_2 \,p^2+ p(q-p)+H \right) \right]  \left( \beta_2 \,q^2 - q(q-p)-H \right)  = \rho  \beta_3 \, p^2 q^2.
\end{equation}
This equivalent formulation of the stability problem has the advantage of the curve obtained being symmetric about the origin, which greatly simplifies the analysis of its singular points (finite or at infinity). We remark that the cases considered in \S~\ref{sec:omega_1_is_0} and \S~\ref{sec:uniform_vorticity} can be recovered from \eqref{curve_gen_case} by setting $\lambda =0$, and $\lambda=1$, respectively.
%, are certainly a good testament to this. 

The homogeneous polynomials of \eqref{curve_gen_case} are found as:
\begin{align}\label{P4_gen_case}
P_4 (p,q)&= pq \Big[ (\lambda +H \beta_1 + \rho (\beta_2-1)) p^2 + \nonumber\\
&\phantom{=pq} + \left(\lambda(\beta_2-2)+ H\beta_1 (\beta_2-1) + \rho (\beta_2-1)^2 + \rho (1-\beta_3) \right)  \,pq + (\beta_2-1) (\rho-\lambda) q^2 \Big], \\
P_2(p,q) &= -H (\lambda+H \beta_1 + \rho (\beta_2-1)) p^2 + H(\lambda-1)\, pq + H (\beta_2-1) (\rho-1) q^2, \nonumber \\
P_0 &= H^2 (1- \rho) \nonumber. 
\end{align}
Consider the quadratic equation for $v=p/q$: 
\begin{equation}\label{quad_form_gen_case}
(\lambda +H \beta_1 + \rho (\beta_2-1)) \,v^2 + (-2 \lambda + 2\rho  -H \beta_1 + \lambda \beta_2 + H \beta_1 \beta_2 -2 \rho \beta_2 + \rho \beta_2^2 - \rho \beta_3 ) \,v +  (\beta_2-1) (\rho-\lambda)=0.  
\end{equation}
Contrary to the cases above, the highest and lowest order coefficients for this quadratic can vanish for specific parameters. As a consequence, $p$ or $q$ may not be simple factors for $P_4$, which could prevent the existence of vertical or horizontal asymptotes for the curve. Namely, using homogeneous coordinates, when $\lambda = \rho$, $p$ becomes a repeated factor of $P_4(p,q)$, with $(0,1,0)$ being a singular point at infinity. Also, if $\lambda = -( H\beta_1 + \rho (\beta_2-1))$, then $q$ is a repeated factor of $P_4 (p,q)$ and a singularity at infinity arises at $(1,0,0)$. 
%(see more details in Appendix E). 

%The two cases will be treated separately, and will be excluded from the analysis that follows. 
To inspect when other singularities at infinity occur, we compute the discriminant of the quadratic in \eqref{quad_form_gen_case}, and obtain: 
\begin{equation}\label{dis_in_lambda}
\beta_2^2 \lambda^2 + d_1 \lambda + d_2,
\end{equation} 
with the following coefficients $d_1$, $d_2$:
\begin{align*}
d_1&= 2 \Big( -H \beta_1 \beta_2 + H \beta_1 \beta_2^2 - 2 \rho \beta_2^2 + \rho \beta_2^3 + 2\rho \beta_3 - \rho \beta_2 \beta_3 \Big),\\
d_2&= H^2 \beta_1^2 (\beta_2-1)^2 + 2 H \beta_1 (\beta_2-1) ( \beta_2^2 - 2\beta_2 - \beta_3)\, \rho + (\beta_2^2-\beta_3)\left( 4-4 \beta_2 + \beta_2^2 - \beta_3 \right) \rho^2.
\end{align*}
Every time \eqref{dis_in_lambda} vanishes, we have a singular point at infinity for the curve \eqref{curve_gen_case}.
It can be shown that \eqref{dis_in_lambda} has two real roots for $\lambda$ regardless of the physical parameters specified. Moreover, as illustrated in Fig.~\ref{transition_config_diagram}, given fixed values of $\rho$, $H$, and $\lambda$ satisfying $\lambda<\rho$ and $\lambda \neq 0$, there are two instances for the wavenumber $\alpha$ ($\alpha = \alpha_l$ and $\alpha = \alpha_r$) at which \eqref{dis_in_lambda} vanishes.
%the discriminant of \eqref{quad_form_gen_case} 
%vanishes, provided $\lambda<\rho$ and $\lambda \neq 0$. 
On this account, four asymptotes for the curve~\eqref{curve_gen_case} are expected, both for small and large values of $\alpha$. 
%the parabola for $\lambda$ in \eqref{dis_in_lambda} is positive, and so 
%four asymptotes for the curve~\eqref{curve_gen_case} are expected. 
For intermediate values of $\alpha$ ($\alpha_l<\alpha<\alpha_r$), however, only the asymptotes $p=0$ and $q=0$ exist. 
%It can also be seen that for small and large values of $\alpha$ the parabola for $\lambda$ in \eqref{dis_in_lambda} is positive, and so four four asymptotes for the curve~\eqref{curve_gen_case} are expected in such cases. 

The analysis of the finite singular points for the curve is considerably more technical, and hence is left to the Appendices. It will be simply noted here that such singularities cannot exist for the curve defined by \eqref{curve_gen_case}, which is 
%Our findings in Appendix D reveal the non-existence of such singularities, which is
of great help, since we may conclude that changes on the topological structure of the curve can only exist through parameters for which the curve has a singular points at infinity (see Appendix D).
\begin{figure}[htbp]
\begin{center}
\includegraphics*[width=340pt]{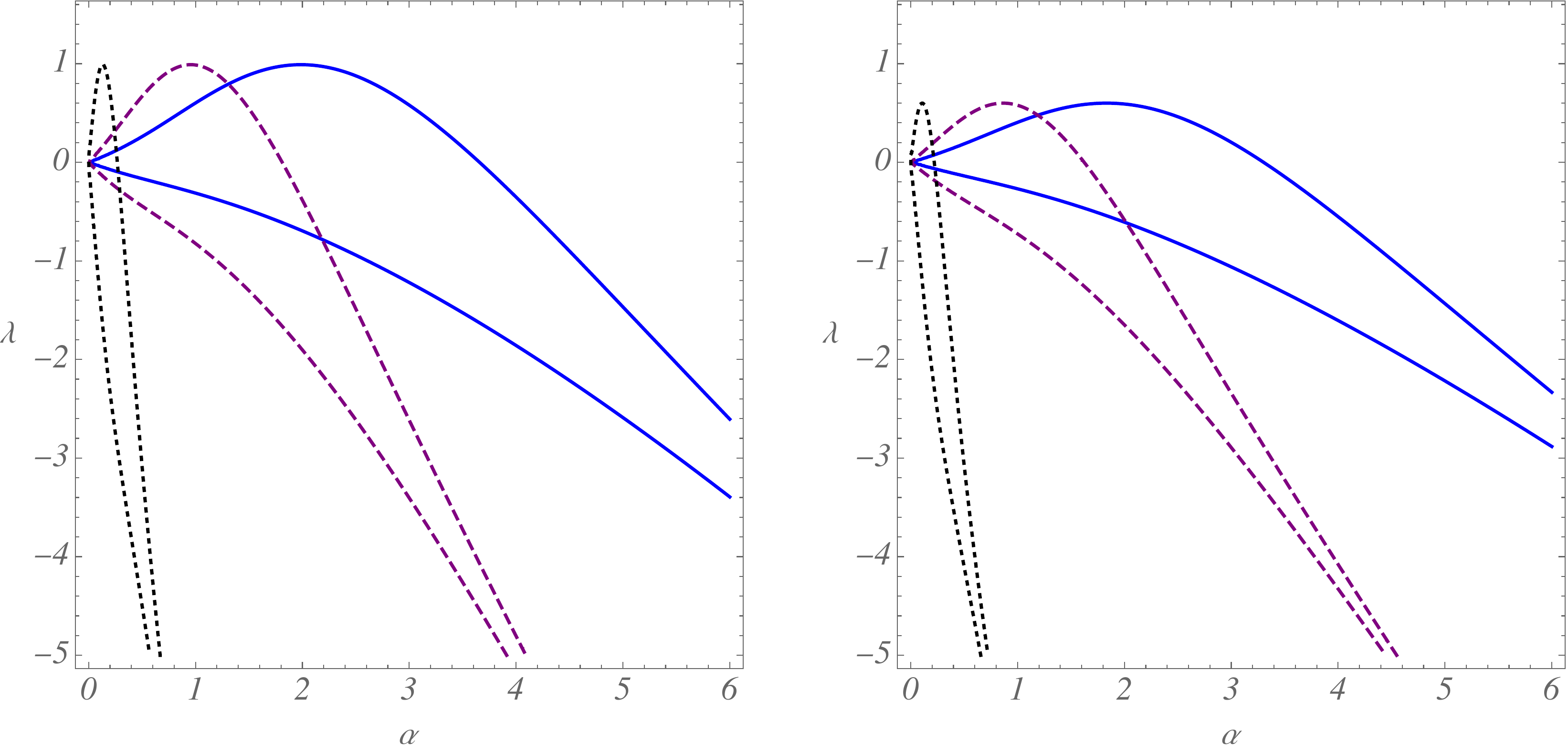}
\\
$(a)$ \hspace{4cm} $(b)$
\end{center}
\caption{Locus in the $(\alpha,\lambda)$-plane where \eqref{dis_in_lambda} vanishes. At these specific locations the curve \eqref{curve_gen_case}  has a singular point at infinity, and a transition of configurations takes place. Different values of $H$ are considered: $H=0.5$ (full line); $H= 1$ (dashed line); $H = 5$ (dotted line). $(a)$ $\rho=0.99$, $(b)$ $\rho=0.6$. In these diagrams, the origin $(0,0)$ always belongs to the geometrical locus where \eqref{dis_in_lambda} vanishes. To account for all cases when singularities at infinity are produced, the lines $\lambda=\rho$ and the graph of $\lambda = -( H\beta_1 + \rho (\beta_2-1))$ should be added to the diagram.
\label{transition_config_diagram}
}
\end{figure}
\begin{figure}[htbp]
\begin{center}
\includegraphics*[width=170pt]{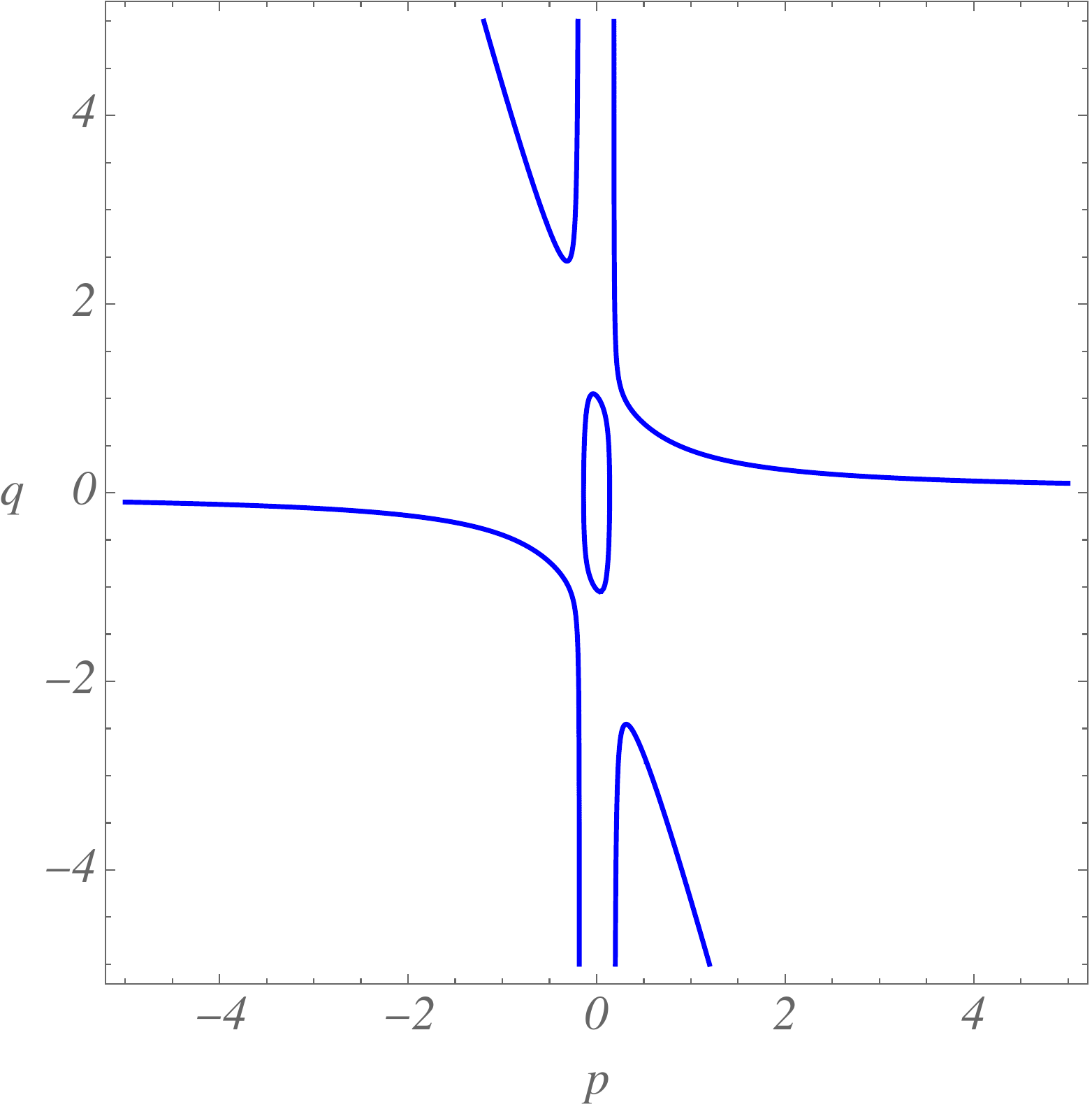}
\end{center}
\caption{Plots on the $(p, q)$-plane of the curve defined by Eq.~\eqref{curve_gen_case} for $\rho=0.9$, $H=0.5$, $\lambda=\rho$, and $\alpha =2.5$. 
\label{lambda_is_rho}
}
\end{figure}

%\begin{figure}[htbp]
%\begin{center}
%\includegraphics*[width=170pt]{figures/transition_config_gen_case_rho_99_100}
%\qquad 
%\includegraphics*[width=170pt]{figures/transition_config_gen_case_rho_6_10}
%\\
%$(a)$ \hspace{4cm} $(b)$
%\end{center}
%\caption{Locus in the $(\alpha,\lambda)$-plane where \eqref{dis_in_lambda} vanishes. At these specific locations the curve \eqref{curve_gen_case}  has a singular point at infinity, and a transition of configurations takes place. Different values of $H$ are considered: $H=0.5$ (full line); $H= 1$ (dashed line); $H = 5$ (dotted line). $(a)$ $\rho=0.99$, $(b)$ $\rho=0.6$. In these diagrams, the origin $(0,0)$ always belongs to the geometrical locus where \eqref{dis_in_lambda} vanishes. To account for all cases when singularities at infinity are produced, the lines $\lambda=\rho$ and the graph of $\lambda = -( H\beta_1 + \rho (\beta_2-1))$ should be added to the diagram.
%\label{transition_config_diagram}
%}
%\end{figure}
%\begin{figure}[htbp]
%\begin{center}
%\includegraphics*[width=170pt]{figures/curve_lambda_is_rho_large_alpha}
%\end{center}
%\caption{Plots on the $(p, q)$-plane of the curve defined by Eq.~\eqref{curve_gen_case} for $\rho=0.9$, $H=0.5$, $\lambda=\rho$, and $\alpha =2.5$. 
%\label{lambda_is_rho}
%}
%\end{figure}

Turning our attention to the case when $\alpha \gg 1$, we can then predict the existence of at least three distinct configurations (with at most three different topological structures) 
%It follows from these results that when $\alpha \gg 1$ there are only three distinct configurations 
for the family of curves \eqref{curve_gen_case}, according to: ($i$) $\lambda>\rho$; ($ii$) $\lambda=\rho$; ($iii$) $\lambda<\rho$. In all three cases, we find curves endowed with four asymptotes and an oval-shaped contour around the origin. Also, according to Appendix E, for each one of the cases, there is one single arrangement of the branches of the curve connecting the asymptotes. While topologically speaking the curves are all the same, their configuration may be distinguished in the following way: 
in the first case, in addition to the vertical and horizontal asymptotes $p=0$ and $q=0$, respectively, there are two oblique asymptotes whose slopes have opposite signs, as in Fig.~\ref{plot_diagram_constant_vorticity}$(b)$; in the latter case, the slope for the two oblique asymptotes have the same (negative) sign, as in Fig.~\ref{plot_irrotational_lower_layer}$(b)$. The transition between the two configurations occurs when $\lambda=\rho$, at which the curve has a node at infinity and two parallel vertical asymptotes exist, {\it cf.} Fig.~\ref{lambda_is_rho}. Clearly, in all cases, instability holds for intermediate values of Froude numbers. As a result, short waves are always unstable, and we may state:
% We may then state the following result, which disproves the assertion made by Curtis {\it et al.} (2017):
%As a result, short waves are always unstable, which by itself disproves the assertion made by Curtis {\it et al.} (2017). We can then state:
\begin{prop}
Consider the two-layer shear flow with piecewise constant vorticity defined by \eqref{physical_configuration} with $\rho_2<\rho_1$. Then the flow is stable to disturbances of arbitrary wavenumber if and only if $\Omega_2 =0$.
\end{prop}

For completeness, it is worth mentioning that although a complete characterization of the configurations obtained for the family of curves \eqref{curve_gen_case} was not pursued here, such investigation would require a close examination of each one of the six subsets in the parameter space listed below, which can easily be identified from Figure~\ref{transition_config_diagram}:  
\begin{equation}
\begin{array}{cccccc}
\text{(I)} & \lambda > \rho, & \text{(II)} & \lambda = \rho, & \text{(III)} & 0<\lambda < \rho,\\
\text{(IV)} & \lambda = 0, & \text{(V)} & -H<\lambda < 0, & \text{(VI)} & \lambda \leqslant -H.
\end{array}\label{six_cases}
\end{equation}
For a particular stratification relevant for real applications ($\rho=0.9$, $H=0.2$), we illustrate in Fig.~\ref{stability_diagrams_full_set} how changes in the value of $\lambda$ can affect the stability features of the flow. Different values of $\lambda$ are prescribed to cover each one of the cases listed in \eqref{six_cases}.     
%In Fig.~\ref{stability_diagrams_full_set} we highlight the effects on the stability features of the flow caused by a change in value the vorticity in each layer, for a particular stratification relevant for real applications ($\rho=0.9$, $H=0.2$).  
%It is important to emphasize that although the parameters $\rho$ and $H$ in Figure~\ref{stability_diagrams_full_set} have been fixed, the statements made below hold in general.
Although the stability diagrams in Figure~\ref{stability_diagrams_full_set} show a great variety, there is typically a single band of wavenumber for instability, for for fixed values of the Richardson number $J$. In panel $(c)$, corresponding to the case (III), two instability bands can, however, be found for small values of $J$. It is interesting to note that in panels $(a)$ and $(c)$, corresponding to the cases (I), (III), respectively, the stability boundaries form a cusp at the origin.  

The figure also suggests that considering negative values $\lambda$, as in cases (V), (VI),
%in which case the vorticities in the upper and lower layer have opposite signs, 
has a stabilizing effect on longer waves. More precisely, it can shown that stability holds in this case for wave numbers below $\alpha_l$, where $\alpha_l$ is the smallest root of \eqref{dis_in_lambda}, when solved for $\alpha$. Furthermore, as the value of $\lambda$ further decreases, the larger is the range of wave numbers for which the stability holds, as shown in panels $(e)$,$(f)$.

\begin{figure}[htbp]
\begin{center}
\includegraphics*[width=430pt]{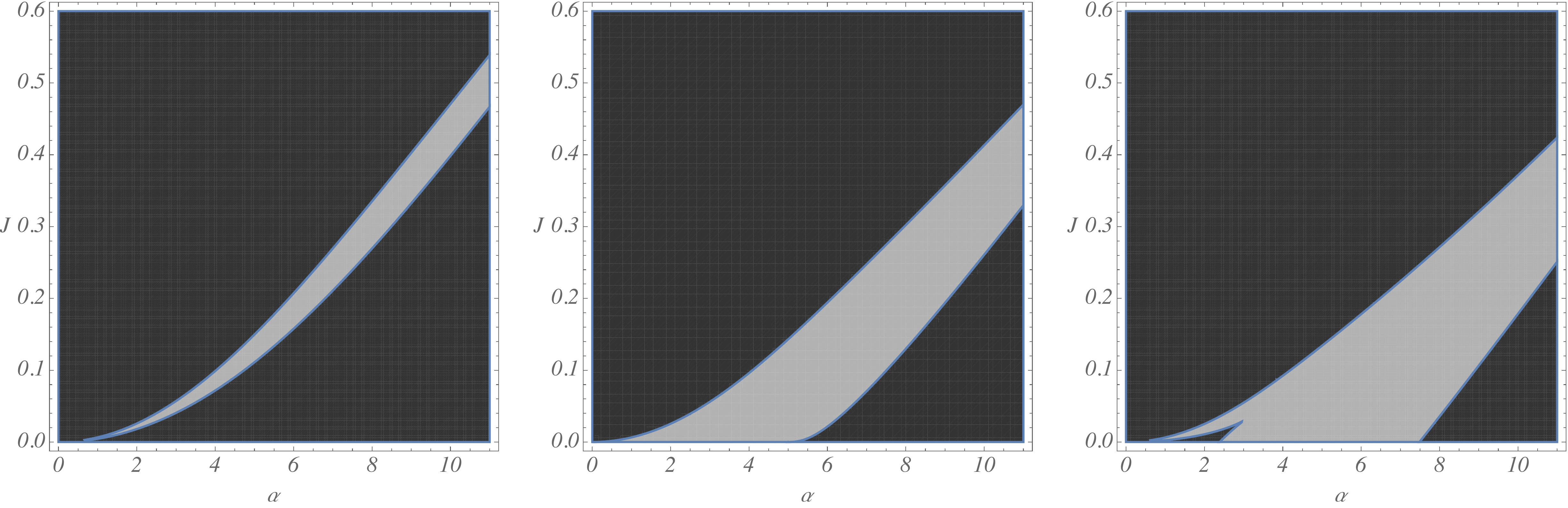}
\\
$(a)$ \hspace{4cm} $(b)$ \hspace{4cm} $(c)$\\
\includegraphics*[width=430pt]{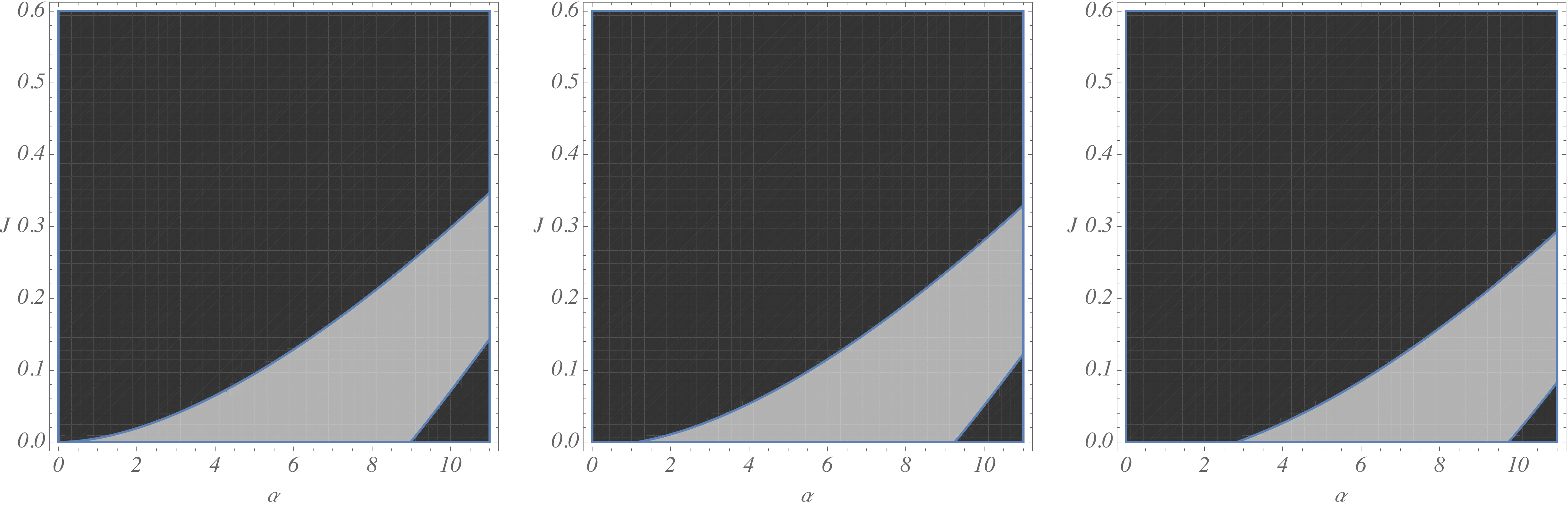}
\\
$(d)$ \hspace{4cm} $(e)$ \hspace{4cm} $(f)$\\

\end{center}
\caption{Stability diagram on the $(\alpha, J)$-plane for the general case given by Eq.~\eqref{curve_gen_case} with $\rho=0.9$, $H=0.2$. Different values of $\lambda$ are considered to cover all the six cases (I)--(VI) given in \eqref{six_cases}. $(a)$ $\lambda = 2$. $(b)$ $\lambda = 0.9$. $(c)$ $\lambda = 0.5$. $(d)$ $\lambda = 0$. $(e)$ $\lambda = -0.1$. $(f)$ $\lambda = -0.3$. Dark shaded: stable region; light shaded: unstable region. In panels $(a)$ and $(c)$ the stability boundaries form a cusp at the origin. 
\label{stability_diagrams_full_set}
}
\end{figure}

\section{A stability criterion for bilinear shear currents in a homogeneous fluid}\label{sec:homogeneous_case}

The results from previous section reveal that stratification in density is insufficient to stabilize bilinear shear currents, and could even promote instability of the flow. Here, we investigate the conditions under which bilinear shear currents in a homogeneous fluid are stable. The case when the shear current is modelled by $N$ constant vorticity layers was recently addressed by Chesnokov, El, Gavrilyuk, and Pavlov~\cite{Chesnokov_et_al}, and sufficient conditions for the long wave limit stability were proposed (see Lemma 3.2 in \cite{Chesnokov_et_al}). 
According to the authors, provided the vorticities are ordered from the bottom to the top layer (and possibly of different sign values), long waves are stable. Remarkably, this same criterion was proposed by Carpenter, Tedford, Heifetz, and Lawrence (see \S 4.2 in \cite{Carpenter_et_al}) for the case when dispersive effects of Euler equations are considered for the piecewise-linear shear current in an infinite domain. 
%when unbounded from below and above.}

Here, dispersive effects and finite-depth effects are considered, but the analysis is limited to bilinear shear currents ($N=2$). As we will see in Proposition~\ref{hom_results}, the geometrical approach allows us to derive necessary and sufficient conditions for the stability of the flow. In contrast with the stability criterion for long waves proposed in \cite{Chesnokov_et_al}, stability to disturbances of arbitrary wavenumber requires vorticities to be ordered, but also have the same sign.

%Here, dispersive effects of Euler equations are considered, but the analysis is limited to bilinear shear currents ($N=2$). 

By taking in \eqref{curve_gen_case} the limit when $\rho \rightarrow 1$, degeneracy occurs and the curve splits up into a vertical line $p=0$ (spurious solution) and a cubic curve. %The vertical line is clearly a spurious solution, resulting from the fact that there is no longer a density interface, but simply a vorticity interface. 
The cubic curve describing the stability problem 
%for the homogeneous piecewise linear shear current 
is given as:
\begin{equation}\label{hom_case}
P_3 (p,q) + P_1(p,q)=0,
\end{equation}
with homogeneous polynomials $P_3$ and $P_1$ defined as follows:
\begin{align*}
P_3(p,q)&= q \Big( (\lambda + H \beta_1 + \beta_2-1)\,  p^2 + (2-H\beta_1 - 2\beta_2 + H \beta_1 \beta_2 + \beta_2^2 - \beta_3 -2\lambda + \lambda \beta_2) \,pq +\\
&\phantom{= q \Big(} +(\beta_2-1) (1-\lambda) \,q^2 \Big), \\
P_1 (p,q) &= -H (\lambda + H \beta_1 + \beta_2-1) p + H (\lambda-1)q.
\end{align*}
When $\lambda=1$ or $\lambda+H\beta_1 + \beta_2-1=0$, the curve with Eq.~\eqref{hom_case} becomes degenerate. 
%A curve is said to be degenerate if it can split up into two or more curves. This is precisely what happens to the curve with Eq.~\eqref{hom_case} when $\lambda=1$ or $\lambda+H\beta_1 + \beta_2-1=0$. 
In the former case, the curve splits up into a vertical line $p=0$ and a hyperbola. Finite singular points may exist, depending on the parameters considered, {\it cf.} Fig.~\ref{hom_case_lambda_is_1}. In the latter case, the curve splits up into a horizontal line $q=0$ and a hyperbola. Here, $(1,0,0)$ becomes a singular point at infinity. In either case, the hyperbola can be parameterized as $p=p(q)$, and any line with slope one as in \eqref{relationship_p_q} cuts the curve in three distinct points. 
\begin{figure}[htbp]
\begin{center}
\includegraphics*[width=340pt]{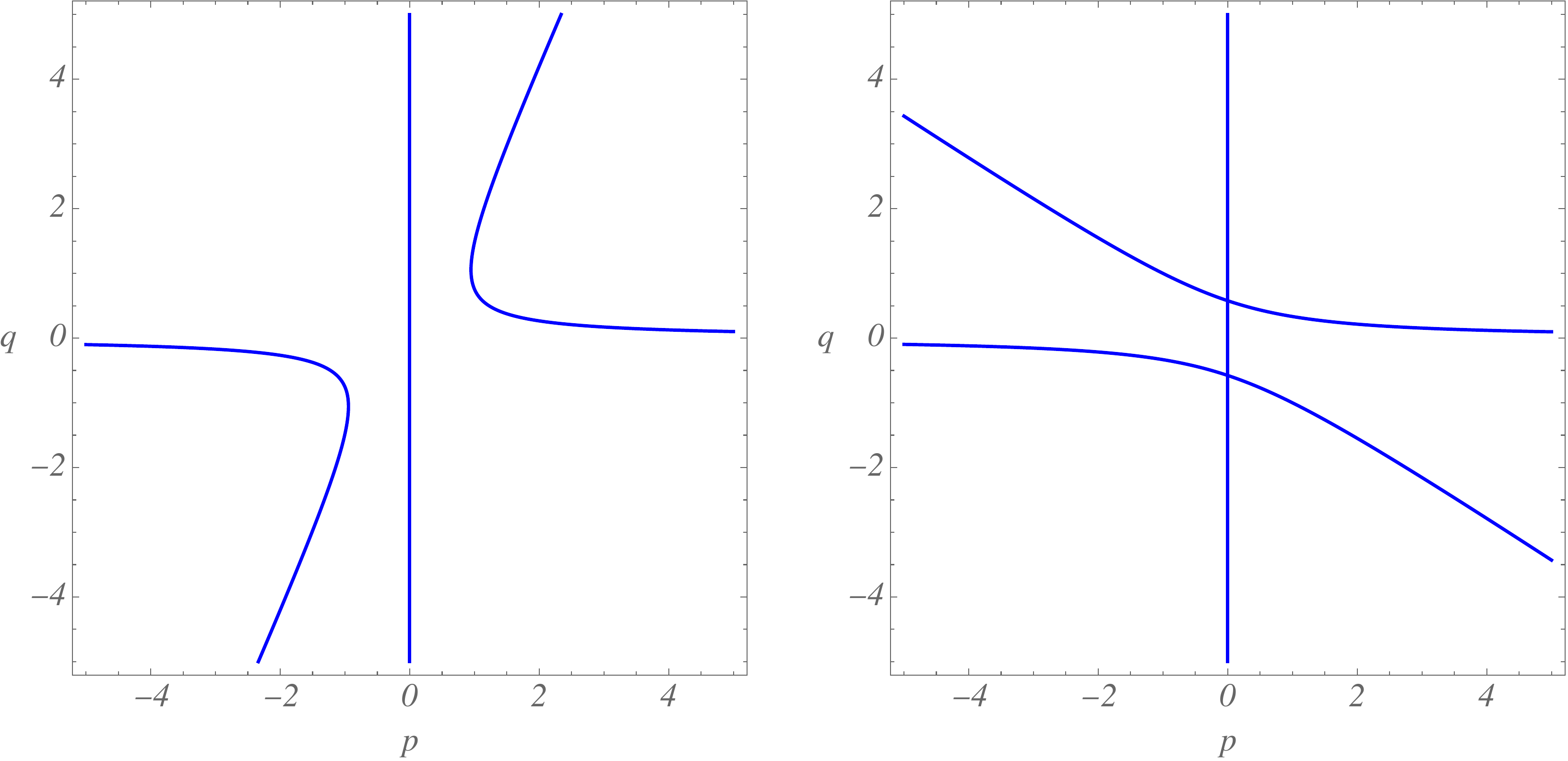}
\\
$(a)$ \hspace{4cm} $(b)$
\end{center}
\caption{Plots on the $(p, q)$-plane of the curve defined by Eq.~\eqref{hom_case} for $H=0.5$ in the degenerate case when $\lambda=1$. $(a)$ $\alpha =1$. $(b)$ $\alpha =5$. 
\label{hom_case_lambda_is_1}
}
\end{figure}

%\begin{figure}[htbp]
%\begin{center}
%\includegraphics*[width=170pt]{figures/hom_uniform_vorticity_small_alpha}
%\qquad 
%\includegraphics*[width=170pt]{figures/hom_uniform_vorticity_large_alpha}
%\\
%$(a)$ \hspace{4cm} $(b)$
%\end{center}
%\caption{Plots on the $(p, q)$-plane of the curve defined by Eq.~\eqref{hom_case} for $H=0.5$ in the degenerate case when $\lambda=1$. $(a)$ $\alpha =1$. $(b)$ $\alpha =5$. 
%\label{hom_case_lambda_is_1}
%}
%\end{figure}

To inspect when other singularities at infinity occur, consider the discriminant of the quadratic obtained from \eqref{quad_form_gen_case} by substituting $\rho$ by 1, providing the repeated factors of $P_3$, other than $q$:
\begin{equation}\label{discriminant_hom_case}
\beta_2^2 \lambda^2 + \tilde{d}_1 \lambda + \tilde{d}_2.
\end{equation} 
Here, the coefficients $\tilde{d}_1$, $\tilde{d}_2$ are obtained, respectively, from $d_1$ and $d_2$ in \eqref{dis_in_lambda} by substituting $\rho$ by 1.  

\begin{figure}[htbp]
\begin{center}
\includegraphics*[width=170pt]{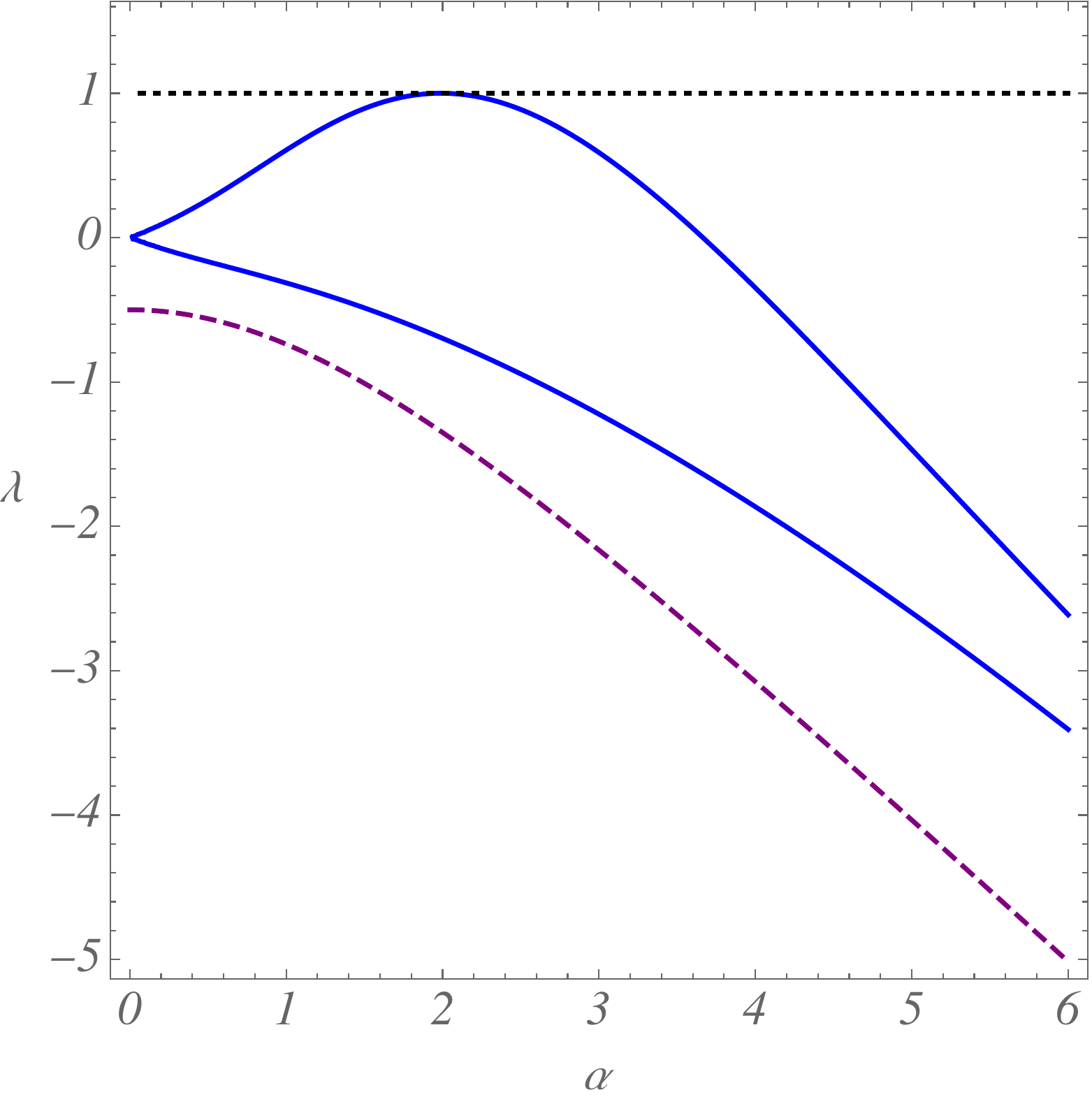}
\end{center}
\caption{Locus in the $(\alpha,\lambda)$-plane where \eqref{discriminant_hom_case} vanishes with $H=0.5$ (full line). The location where the curve becomes degenerate is given by $\lambda = -( H\beta_1 + \beta_2-1)$ (dashed line) and $\lambda=1$ (dotted line).
\label{transition_config_diagram_hom_case}
}
\end{figure}

%\begin{figure}[htbp]
%\begin{center}
%\includegraphics*[width=170pt]{figures/transition_config_hom_case_H_5_10}
%\end{center}
%\caption{Locus in the $(\alpha,\lambda)$-plane where \eqref{discriminant_hom_case} vanishes with $H=0.5$ (full line). The location where the curve becomes degenerate is given by $\lambda = -( H\beta_1 + \beta_2-1)$ (dashed line) and $\lambda=1$ (dotted line).
%\label{transition_config_diagram_hom_case}
%}
%\end{figure}

As shown in Appendix D, finite singular points cannot exist unless $\lambda=1$. We can then identify the subregions in the parameter space where no singularities occur (see Fig.~\ref{transition_config_diagram_hom_case}), and try to characterize in each one of these subregions how many possible curve configurations there may be. In particular, if one focus strictly on short waves ($\alpha \gg 1$), we can foresee the existence of at least three distinct configurations for the the family of curves given by Eq.~\eqref{hom_case}, according to: ($i$) $\lambda>1$; ($ii$) $\lambda=1$; ($iii$) $\lambda<1$. The transition between the configurations in $(i)$ and $(iii)$ (see Fig.~\ref{hom_case_large_alpha}) occurs when $\lambda=1$, at which the curve is degenerate and has two finite singular points. 
Similarly to what is described in Appendix E for the stratified case, it can be shown that, provided $\alpha$ is large enough, for each subregion of the parameter space ($\lambda>1$, or $\lambda<1$) there is one single arrangement of the curve branches connecting the asymptotes. Therefore, if the curve is non-degenerate and the values of $\alpha$ are large enough, there are strictly two distinct configurations for the curve, {\it cf.} Fig.~\ref{hom_case_large_alpha}.

Clearly, in the case $(iii)$, illustrated in Fig.~\ref{hom_case_large_alpha}($b$), instability holds for a finite range of Froude numbers. In case $(i)$, however, stability holds, according to Fig.~\ref{hom_case_large_alpha}($a$). More importantly, it can be shown that this configuration, obtained for $\lambda>1$ and large values of $\alpha$, persists for any strictly positive value of $\alpha$, even if arbitrarily small.  
%and $(ii)$, however, stability holds. 
Hence we have:

\begin{figure}[htbp]
\begin{center}
\includegraphics*[width=340pt]{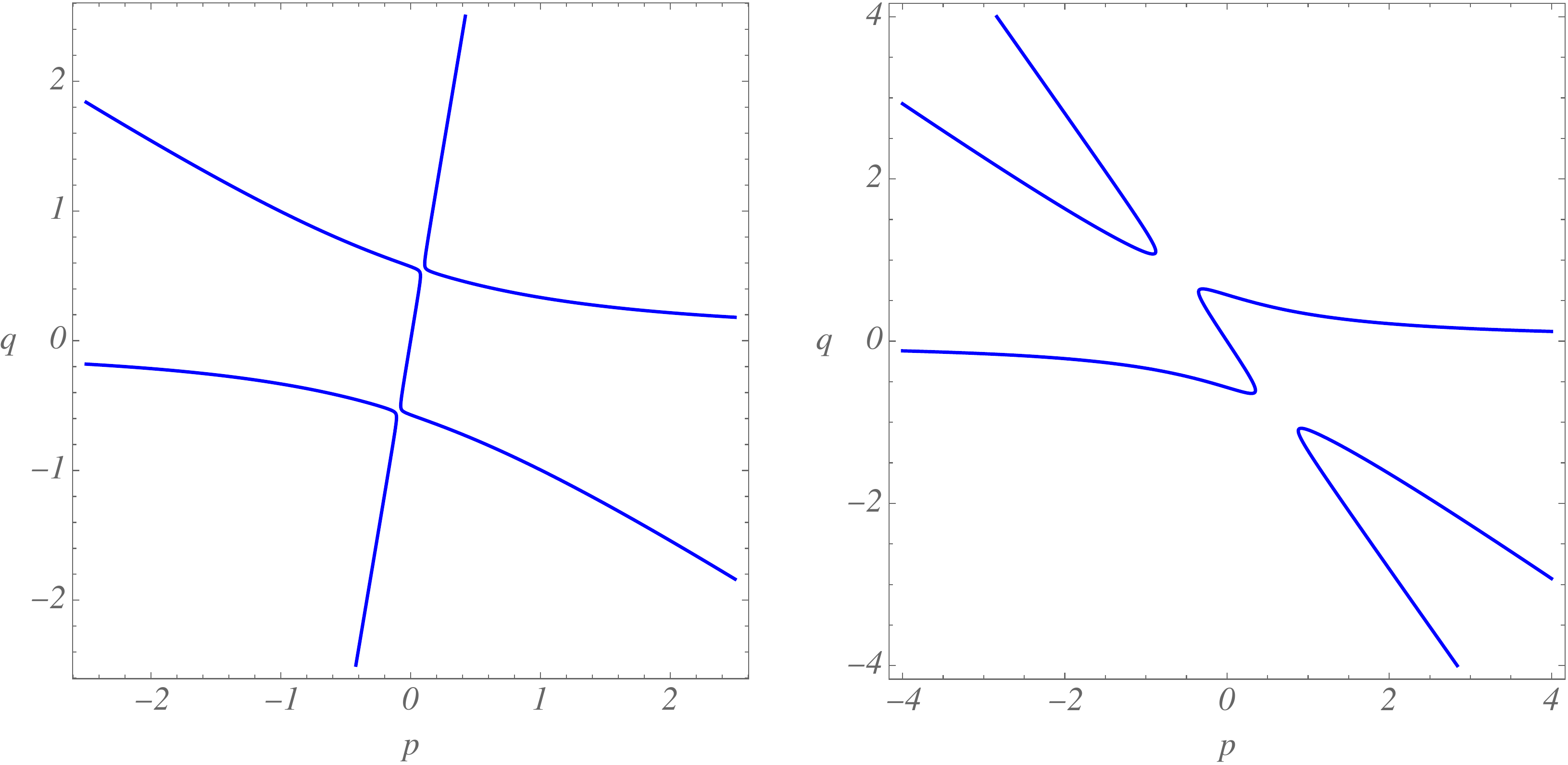}
\\
$(a)$ \hspace{4cm} $(b)$
\end{center}
\caption{Plots on the $(p, q)$-plane of the curve defined by Eq.~\eqref{hom_case} for $H=0.5$ and $\alpha=5$, with different values of $\lambda$. $(a)$ $\lambda =2$. $(b)$ $\lambda =-1$. These two configurations are the only possible configurations for non-degenerate curves of the family with Eq.~\eqref{hom_case}, provided the values of $\alpha$ are large enough.
\label{hom_case_large_alpha}
}
\end{figure}

%\begin{figure}[htbp]
%\begin{center}
%\includegraphics*[width=170pt]{figures/large_alpha_lambda_greater_1}
%\qquad 
%\includegraphics*[width=170pt]{figures/large_alpha_lambda_less_1}
%\\
%$(a)$ \hspace{4cm} $(b)$
%\end{center}
%\caption{Plots on the $(p, q)$-plane of the curve defined by Eq.~\eqref{hom_case} for $H=0.5$ and $\alpha=5$, with different values of $\lambda$. $(a)$ $\lambda =2$. $(b)$ $\lambda =-1$. These two configurations are the only possible configurations for non-degenerate curves of the family with Eq.~\eqref{hom_case}, provided the values of $\alpha$ are large enough.
%\label{hom_case_large_alpha}
%}
%\end{figure}

\begin{prop}\label{hom_results}
Consider the homogeneous shear flow with piecewise constant vorticity defined by \eqref{physical_configuration} with $\rho_1=\rho_2$. Then the flow is stable to disturbances of arbitrary wavenumber if and only if 
$$
0\leqslant\Omega_2 \leqslant \Omega_1, \quad {\text or} \quad 0\geqslant\Omega_2 \geqslant \Omega_1.
$$
\end{prop}

\section{Concluding remarks}

We have found necessary and sufficient conditions for the stability of a free surface fluid current modeled as two finite layers of constant vorticity. The criteria depend on the strength, but not on the thickness, of each vorticity layer. 
According to the Remark 3.4 in \cite{Chesnokov_et_al}, this would not be the case of a basic flow composed by three finite layers of constant vorticity, for which the stability of long waves may also depend on the values of depth ratios.

Although the physical setting here (see Fig.~\ref{sketch}$(b)$) was considered in previous studies (see {\it e.g.}{ Dalrymple~\cite{Dalrymple}), the question of stability had not been fully addressed. Our results indicate that caution should be taken when modeling a fluid current as a bilinear shear current. Take for example the case when the mean horizontal velocity $U(z)$ is convex and monotonically increasing. According to Yih~\cite{Yih}, such flow is stable.
%In particular, if the mean horizontal velocity $U(z)$ is convex and monotonically increasing, in which case the flow is stable, according to Yih~\cite{Yih},
%Hur and Lin~\cite{Hur_Lin}, 
Any attempt to model this current as a bilinear shear current could thus lead to important discrepancies with experimental results, since such reduction would lead to an unstable configuration.  
%where the wave speeds within certain propagation modes become complex. 
The extent of the consequences of doing so with many more vorticity layers (as in Ref.~\cite{Swan_et_al}) is currently not known, as only sufficient conditions are known for the stability of the flow in the longwave limit \cite{Chesnokov_et_al}, and warrants further study. 

Comparison with the setting where the vorticity interface becomes also a density interface, separating two immiscible liquids with constant densities, shows that stratification in density has a strong destabilizing effect on the stability of the flow. It is found that only when the vorticity of the upper layer vanishes, the flow is stable, 
%regardless of the other physical parameters considered, 
which is in disagreement with the claim made by Curtis, Oliveras, and Morrison~\cite{Curtis_et_al} that suppression of instability for weak upper layer vorticity is a generic feature.  
%made by Curtis, Oliveras, and Morrison~\cite{Curtis_et_al}. 
Their erroneous conclusion was based on the numerical results to inspect the nature of roots for the dispersion relation (see Figs.~3,4 in \cite{Curtis_et_al}) and asymptotic solutions in the Boussinesq regime. The arguments are misleading for the following reasons: when the vorticity in the upper layer is weak ({\it i.e.,} $F\approx 0$), the instability manifests only at very short waves ($\alpha\gg1$). Therefore, given that the instability band becomes increasingly narrower as the wavenumber $\alpha$ increases (see Figs.~\ref{stability_diagram_irrotational_lower_layer},\ref{plot_diagram_constant_vorticity}), it can hardly be detected numerically. In addition, the complex roots within this instability band have imaginary parts that are very close to zero. Hence, when using asymptotics, if the polynomial roots are determined only at leading order, as in \cite{Curtis_et_al}, one fails to predict their complex nature.  

It would be interesting to explore the role of these instabilities on the stability properties of finite-amplitude waves in these physical settings. 
%In the case of a homogeneous fluid, evidence for the persistence of linear instability of the background flow 
This will be reserved to future work.
We conclude my mentioning that all the rigorous results in this study, summarized in Table~\ref{table_results}, are valid for idealized models where the physical effects caused by viscosity, surface tension, and three-dimensional motions, are neglected. For real applications the possible inadequacies of such models need to be tested.
%\tcb{For convenience, all results of this study are summarized in Table~\ref{table_results}.}
%, but we are unaware of experiments along these lines.
\begin{table}[htbp]
\centering
\caption{Stability criteria for the physical configurations under consideration.}
\label{table_results}
\includegraphics*[width=370pt]{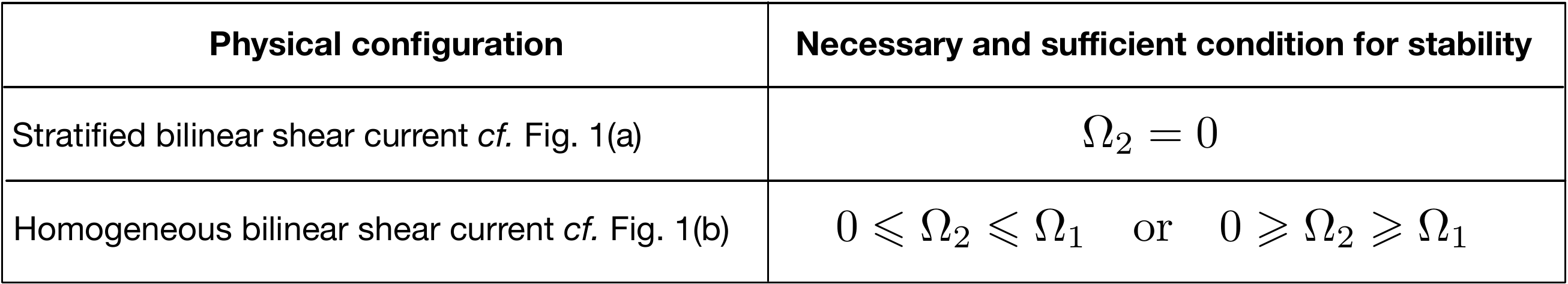}
\end{table}

%\begin{table}[htbp]
%\centering
%\caption{\tcb{Stability criteria for the physical configurations under consideration.}}
%\label{table_results}
%\includegraphics*[width=370pt]{figures/table_results}
%\end{table}

%is insufficient to stabilize bilinear shear currents, and could even promote instability of the flow
%   
%The results from previous section reveal that stratification in density is insufficient to stabilize bilinear shear currents, and could even promote instability of the flow. Here, we investigate the conditions under which bilinear shear currents in a homogeneous fluid are stable.
%shown that stratification in density has a strong destabilizing effect on the stability of bilinear shear currents with a free surface. 

\section*{Acknowledgements}
R.B. would like to thank F. Cl\'{e}ry and H. Ahmadinezhad for helpful discussions. J.F.V. gratefully acknowledges support through the Marsden Fund Council administered by the Royal Society of New Zealand. We thank the anonymous referees for their valuable comments, which helped us improve the quality of this manuscript.

\secta{Appendix A. Nature of roots for the dispersion relation \eqref{dispersion_relation}}
%\secta{Appendix A. Proof that the dispersion relation \eqref{dispersion_relation} cannot have four complex roots}

%The curve \eqref{curve} cuts the axes at four points.

Here, we show that four complex solutions to the eigenvalue equation \eqref{dispersion_relation} can never exist. This feature stems from the geometrical approach in a very natural way, given that any straight line with slope one, as in \eqref{relationship_p_q}, intersects the curve with equation \eqref{curve} at least twice. To prove this, first we establish that $p=0$ and $q=0$ are asymptotes to the curve, regardless of the physical parameters used. Then, we will show that these are connected by the curve one to the other, entirely within the first quadrant (similarly in the third quadrant), from which the result follows.  

We start by casting the curve equation \eqref{curve} into the form
\begin{equation*}\label{homogeneous_polynomials_expression}
P_4 (p,q) + P_3 (p,q) + P_2 (p,q) + P_1 (p,q) + P_0 =0,
\end{equation*}
where each $P_k (p,q)$ is a homogeneous polynomial in $p$ and $q$ of degree $k$, with $P_4$ being defined as in \eqref{P4_exp}, and $P_3$ given as $P_3 = -H \widetilde{\Omega}_1 \,pq \left[p + (\beta_2-1)q\right]$. 
%top degree polynomials $P_3$ and $P_4$ are given by: 
%\begin{align*}
%P_4 &= pq \left[ (H \beta_1 + \rho (\beta_2 -1))\, p^2  + ( -H \beta_1 + H \beta_1 \beta_2 + 2\rho - 2\rho \beta_2 + \rho \beta_2^2-\rho \beta_3) \,p q + \rho(\beta_2-1)\, q^2 \right], \\\
%P_3 &= -H \widetilde{\Omega}_1 \,pq \left[p + (\beta_2-1)q\right].
%\end{align*} 
Unless $\widetilde{\Omega}_1 = 0$ (as in \S~\ref{sec:omega_1_is_0}), odd degree homogeneous polynomials $P_3$ and $P_1$ will be present in the curve expression. To determine its asymptotes, we make use of the result (Primrose~\cite{Primrose}, Theorem 2, pp.~7--8): {\it If $ap + bq$ is a simple factor of $P_4(p,q)$, i.e., if $P_4(p,q) = (ap+bq) \, Q (p,q)$ with $Q(b,-a)\neq 0$, then associated with this factor is the single asymptote to $P(p,q)=0$ defined by the equation
$(ap+bq) \,Q(b,-a) +P_{3}(b,-a)=0.$}
Both $p$ and $q$ are simple factors of $P_4(p,q)$, so each one will have associated an asymptote to the curve \eqref{curve}. Moreover, since $pq$ divides both $P_4$ and $P_3$, we have $P_3 (0,\bullet)=P_3(\bullet,0)=0$, and hence $p=0$ and $q=0$ are asymptotes to the curve \eqref{curve}. 
\begin{figure}[htbp]
\begin{center}
\includegraphics*[width=170pt]{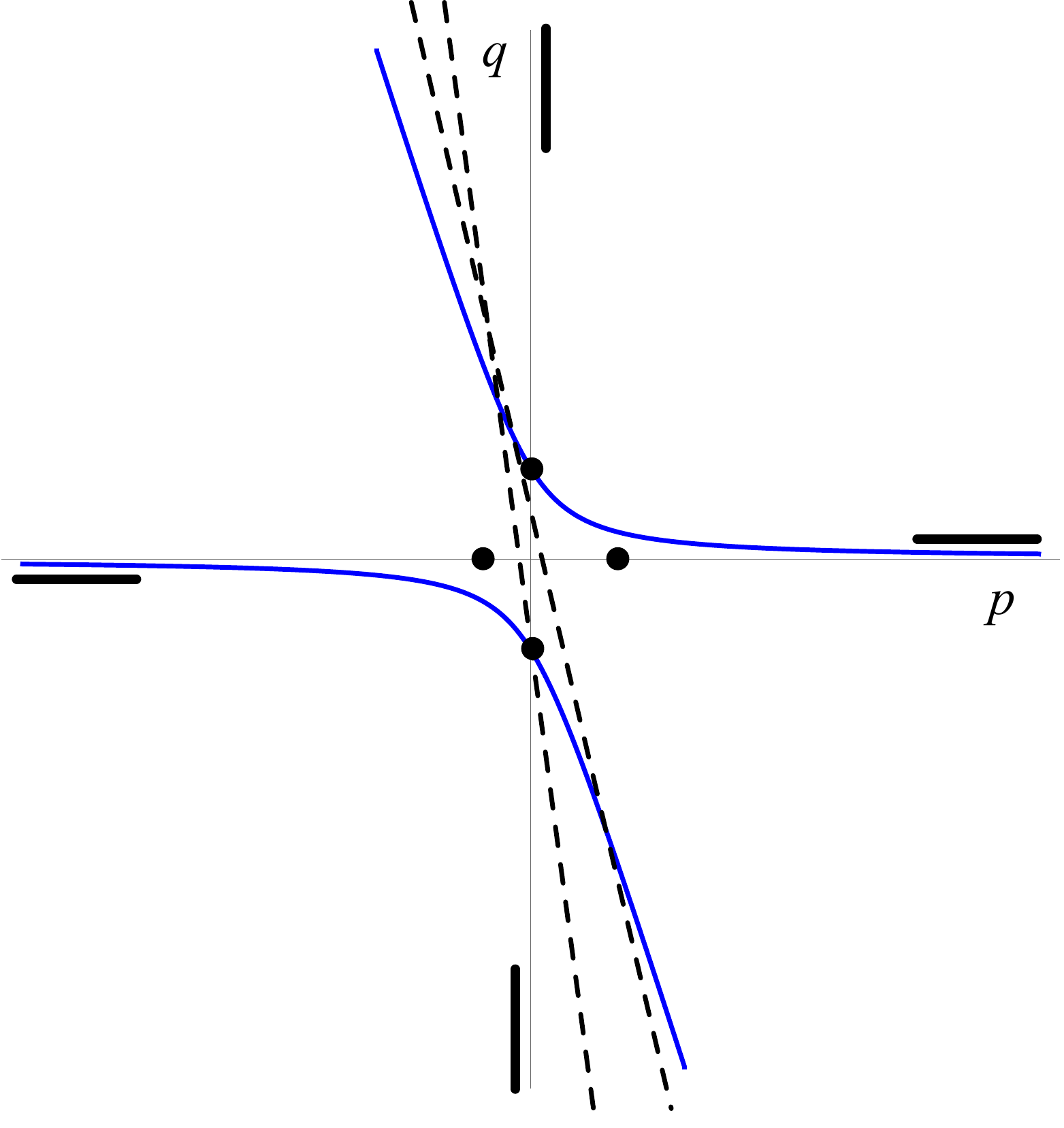}
\end{center}
\caption{Schematic diagram of the curve defined by Eq.~\eqref{curve} illustrating the way the curve tends to infinity, when approaching the asymptotes $p=0$ and $q=0$, and the points where it cuts the axes. Also included is the plot of the hyperbola \eqref{hyperbola_from_degeneracy} tangent to the curve at the $q$-intercepts. The branches of the curve connecting the asymptotes in the first and second (third and fourth) quadrants, are above (below) the hyperbola.
%In the case when the curve has extra asymptotes (dashed line), the branches of the curve in second (fourth) quadrant are above (below) the hyperbola.   
\label{curve_sketch}
}
\end{figure}

%\begin{figure}[htbp]
%\begin{center}
%\includegraphics*[width=170pt]{figures/rough_diagram01}
%%\qquad 
%%\includegraphics*[width=150pt]{figures/rough_diagram02}
%%%\quad 
%%%\includegraphics*[width=150pt]{figures/curve_alpha_large}
%%\\
%%$(a)$ \hspace{4cm} $(b)$
%\end{center}
%\caption{Schematic diagram of the curve defined by Eq.~\eqref{curve} illustrating the way the curve tends to infinity, when approaching the asymptotes $p=0$ and $q=0$, and the points where it cuts the axes. Also included is the plot of the hyperbola \eqref{hyperbola_from_degeneracy} tangent to the curve at the $q$-intercepts. The branches of the curve connecting the asymptotes in the first and second (third and fourth) quadrants, are above (below) the hyperbola.
%%In the case when the curve has extra asymptotes (dashed line), the branches of the curve in second (fourth) quadrant are above (below) the hyperbola.   
%\label{curve_sketch}
%}
%\end{figure}

Although this result gives the asymptotes $p=0$ and $q=0$ quite easily, it does not tell us on which side of the asymptote the curve lies at each end. With regards to the horizontal asymptote $q=0$, it can be shown that, for sufficiently large values of $p$, the equation of the curve can be expressed in the form: 
$$
q= H/p + O(1/p^3). 
$$
Since $H>0$, when $p$ is large and positive, values of $q$ for the curve are greater than those on the asymptote. Hence the curve lies above the asymptote at that end. When $p$ is large and negative, the reverse situation occurs. 

For the asymptote $p=0$ we can consider, for sufficiently large values of $q$, the following expansion for the curve:
$$
p=(1/\rho-1)H/q+O(1/q^2).
$$
When $q$ is large and positive, the values of $p$ for the curve are greater than those on the asymptote. Hence the curve lies to the right of the asymptote at that end. The reverse situation occurs at the other end. 
%These results are summarized in Figure~\ref{curve_sketch}.

It remains to show how the asymptotes are connected. Consider the hyperbola defined by 
\begin{equation}\label{hyperbola_from_degeneracy}
q^2 (1-\beta_2) - pq + H=0,
\end{equation}
appearing in the limit of \eqref{curve} when $\rho$ goes to zero (see Appendix B). When approaching the asymptote $q=0$, the hyperbola admits the following expansion: 
$$
q=H/p + H^2(1-\beta_2)/p^3 + O(1/p^4).
$$
We can then find the difference between the values of the ordinates for the hyperbola and the curve as: 
$$
q_\text{hyp}-q_\text{curve} = - \frac{\rho H^2 \beta_3}{H\beta_1 + \rho (\beta_2 -1)} \frac{1}{p^3} + O(1/p^4).
$$
When approaching the asymptote $q=0$ through large positive values of $p$, we infer that the values of $q$ for the hyperbola are less than those for the curve, and greater than those on the asymptote. 
%For the asymptote $q=0$, we can infer that when $p$ is large and positive, values of $q$ for the hyperbola are less than those for the curve, and greater than those on the asymptote. 
When $p$ is large and negative, values of $q$ for the hyperbola are greater than those for the curve, and less than those on the asymptote.

In addition, it can be proved that, in a non-degenerate case, the hyperbola and the curve 
have two points in common, with coordinates $\left(0,\pm \sqrt{H/(\beta_2-1)}\right)$, at which they are tangent. 
%have two points in common, with coordinates $\left(0,\pm \sqrt{H/(\beta_2-1)}\right)$. It can be proved that at such points, where the curve cuts the $q$-axis, the curve and the hyperbola are tangent one to the other. 
If $p=0$ and $q=0$ are the only asymptotes, then since the curve never crosses the hyperbola 
%(which divides the $(p,q)$-plane in three distinct connected components), 
we conclude that in the first quadrant the asymptotes $q=0$ and $p=0$ are connected through the curve entirely within the quadrant, similarly for the asymptotes in the third quadrant, 
%and third quadrants the asymptotes $q=0$ and $p=0$ must be connected by the curve one to the other, 
{\it cf.} Fig.~\ref{plot_irrotational_lower_layer}. 
If extra asymptotes exist, it can be shown that their slopes, in absolute value, exceed $1/(\beta_2-1)$, which is the absolute value of the oblique asymptote for the hyperbola. As a consequence, the branches of the curve in second (fourth) quadrant connecting the asymptotes are above (below) the hyperbola. All these results are summarized in Fig.~\ref{curve_sketch}.

As a final remark we point out that any rearrangement of the curve branches connecting the asymptotes would lead to finite singularities at the $q$-intercepts. Although, this can never occur, since $P_p (0,\pm \sqrt{H/(\beta_2-1)})= P_q (0,\pm \sqrt{H/(\beta_2-1)})=\pm H (\rho-1) \sqrt{H/(\beta_2-1)} \neq 0$.

In conclusion, any line with with equation \eqref{relationship_p_q} intersects the curve with equation \eqref{curve} at least twice. As such, the wave speeds for the dispersion relation \eqref{dispersion_relation} cannot all be complex. 
%
%To seek intersection points between the curve $P(p,q)=0$ and the hyperbola, we take the resultant between the two equations, with respect to $p$, which yields
%$$
%-\rho \beta_3 \,q^3 \left[ (\beta_2 -1)q^2 -H\right]^2.
%$$
%Since the hyperbola never intersects the $p$-axis (i.e., $q=0$), we conclude that the hyperbola and the curve have  at least one coincident component in a degenerate case (see Appendix B), or two points in common, with coordinates $\left(0,\pm \sqrt{H/(\beta_2-1)}\right)$. It can be proved that at such points, where the curve cuts the $q$-axis, the curve and the hyperbola are tangent one to the other. Since the curve never crosses the hyperbola (which divides the $(p,q)$-plane in three distinct connected components), we conclude that in the first and third quadrants the asymptotes $q=0$ and $p=0$ must be connected by the curve one to the other, {\it cf.} Figure~\ref{plot_irrotational_lower_layer}.
%

\sectb{Appendix B. Degenerate cases}

The form in which we have presented the curve equation~\eqref{curve} suggests some limit cases to be considered: $\rho\rightarrow 0$, $\beta_3 \rightarrow 0$. 
An algebraic curve is said to be degenerate if it can split up into two or more curves. We will see below that this is precisely what happens in each one of the cases just mentioned. \\

\noindent {\bf Case $\bm{\rho \rightarrow 0.}$} 
In this case, the curve~\eqref{curve} can be factorized as
$$
\left( q^2 (1-\beta_2) - pq + H \right) \left( \beta_1 p^2 - \widetilde{\Omega}_1 p -1 \right)=0.
$$
We have the union between two vertical lines and a hyperbola. Stability holds for arbitrary wave number regardless of the physical parameters and the stability result for a homogeneous flow with constant vorticity is recovered.\\

\noindent {\bf Case $\bm{\beta_3 \rightarrow 0.}$}
This limit corresponds to letting $\alpha$ go to infinity. For this reason, it is convenient first to rewrite \eqref{curve} as 
$$
\frac{1}{\alpha} \left[ H (\beta_1 p^2-\widetilde{\Omega}_1 p -1) +\rho(\beta_2 p^2+ p(q-p) +H) \right] \, \frac{1}{\alpha} 
\left( \beta_2 \, q^2 - q (q-p) -H \right) = \rho  H^2 \, \csch^2 (H \alpha) \, \, p^2 q^2.
$$
Then, in the limit when $\alpha \rightarrow \infty$, we get 
$$
p^2 \, q^2 =0,
$$ 
whose intersection with any line with slope $1$ is composed by four points corresponding to two distinct real roots, each with multiplicity two. This confirms
that the instability band in the diagram of Fig.~\ref{stability_diagram_irrotational_lower_layer} is further reduced as $\alpha$ increases.

Similar considerations apply to the family of curves given by Eq.~\eqref{curve_gen_case}. In either case, there is another limit under which the curve becomes degenerate, consisting on letting $\rho$ go to one, from which we recover the results for a homogeneous piecewise linear shear current, covered in \S~\ref{sec:homogeneous_case}.

\sectc{Appendix C. Number of real roots of a quadratic form within a fixed interval}

Consider the polynomials $f(x)=a_0 x^2 + a_1 x+a_2$ ($a_0\neq0$) and $f_1(x)=f'(x)$. Let us seek the greatest common divisor of $f$ and $f_1$ with the help of Euclid's algorithm: 
$$
f= q_1 f_1 - f_2
$$
with $q_1 = (x + a_1/{2a_0})/2$ and $f_2 = (a_1^2-4 a_0 a_2)/{4a_0}$. 
The sequence $f$, $f_1$, $f_2$ is called the Sturm sequence of the polynomial $f$. The following Theorem by Sturm \cite{Prasolov} states the number of real roots of $f$ within a given interval $]a,b[$.

\begin{theo}
Let $\omega(x)$ be the number of sign changes in the sequence
$$
f(x), \quad f_1(x), \quad f_2.
$$
The number of the roots of $f$ (without taking multiplicities into account) confined between $a$ and $b$, where $f(a)\neq 0$, $f(b)\neq 0$ and $a<b$, is equal to $\omega(a)-\omega(b)$.
\end{theo}

%\begin{cor}\label{no_roots_interval}
%Suppose that the discriminant of $f$ is positive, $f(0)$ and $f(1)$ are positive, and the coefficients $a_0$ and $a_1$ are such that $a_0>0$ implies $a_1>0$. Then $f$ has no roots within the interval $]0,1[$.
%\end{cor}
\begin{cor}\label{no_roots_interval}
Suppose that $a_1^2-4 a_0 a_2>0$, $f(0)>0$ and $f(1)>0$, and the coefficients $a_0$ and $a_1$ are such that $a_0>0$ implies $a_1>0$. Then $f$ has no roots within the interval $]0,1[$.
\end{cor}

\begin{proof}
To determine the values $\omega(0)$ and $\omega(1)$ we consider the sequences 
$$
f(0)=a_2, \quad f_1(0) = a_1, \quad f_2,
$$
and 
$$
f(1)=a_0+a_1+a_2, \quad f_1(1)=2a_0+a_1, \quad f_2.
$$
If $a_0>0$, then by hypothesis $a_1>0$ and $f_2>0$. Hence, $\omega(0)=0=\omega(1)$, and no roots can be found between $]0,1[$. 
If $a_0<0$, then $\omega(0)=1$, regardless of the sign of $a_1$. Similarly, $a_0<0$ implies $\omega(1)=1.$ Therefore, the number of real roots of $f$ within $]0,1[$ is given by $\omega(0)-\omega(1)=0$.  
\end{proof} 

\sectd{Appendix D. Finite singular points for the curve \eqref{curve_gen_case}}

To investigate the existence of finite singular points for the family of curves given by \eqref{curve_gen_case}, again we adopt the system of equations \eqref{cand_singular_pts}: 
\begin{equation*}\label{cand_singular_pts_gen}
2 P_0 P_{4,p} - P_2 P_{2,p} =0, \quad 2 P_0 P_{4,q} - P_2 P_{2,q} =0.
\end{equation*}
Using the same strategy as before, we introduce the variable $v=p/q$ and seek common roots for two cubics:
\begin{equation}\label{cubic_one}
2 (\lambda + H \beta_1 + \rho(\beta_2-1))\, v^3 -3 (\lambda + 1-2\rho) (\lambda + H \beta_1 + \rho(\beta_2-1)) \, v^2 + \gamma_{1,1} \, v + \gamma_{1,2} =0, 
\end{equation}
\begin{equation}\label{cubic_two}
(\lambda+1-2\rho) (\lambda + H \beta_1 + \rho(\beta_2-1)) \, v^3 + \gamma_{2,1} \, v^2 + 3(\rho-1) (\beta_2-1) (\lambda + 1-2\rho) \, v + \gamma_{2,2}=0,
\end{equation}
with coefficients $\gamma_{i,k}$ ($i,k=1,2$) dependent on the physical parameters, which for simplicity will be omitted here.
The resultant of \eqref{cubic_one} and \eqref{cubic_two} is given by: 
\begin{equation}\label{resultant_gen_case}
-64 (\beta_2-1)^2 \beta_3^2 (\rho-1)^4 \rho^2 (\lambda + H \beta_1 + \rho (\beta_2-1))^2 G(\rho, \alpha, H,\lambda),
\end{equation}
where $G$ is defined as a quartic for $\lambda$:
\begin{equation}\label{quartic_for_lambda}
G = \lambda^4 +e_1 \lambda^3+e_2 \lambda^2 +e_3 \lambda +e_4,
\end{equation}
with coefficients $e_i$ ($i=1,2,3,4$) listed below: 
\begin{flalign*}
\phantom{e_1}e_1 = -4\Big( 1-2 \beta_2 + 2 \beta_2 \rho \Big),&&
\end{flalign*}
\begin{multline*}
 e_2 = -2\Big(-3+ 4 H \beta_1 + 8 \beta_2 - 4H \beta_1 \beta_2 -8 \beta_2^2) + (-4H \beta_1 + 4 H \beta_1 \beta_2 + 12 \beta_2^2 - 4 \beta_3) \rho + \\
 +(-8\beta_2 - 4\beta_2^2 + 4 \beta_3) \rho^2 \Big), 
\end{multline*}
\begin{multline*}
e_3 = 4 \Big(  -1 + 4 H \beta_1 + 2 \beta_2 -12 H \beta_1 \beta_2 + 8H \beta_1 \beta_2^2 + \\
+(-4H\beta_1+6 \beta_2 +20 H \beta_1 \beta_2 - 20\beta_2^2 -16 H \beta_1 \beta_2^2 + 8 \beta_2^3 + 12 \beta_3 -8 \beta_2 \beta_3) \,\rho +\\
+ \left(-8\beta_2 - 8 H \beta_1 \beta_2 +36 \beta_2^2 + 8H \beta_1 \beta_2^2 -16 \beta_2^3 -28 \beta_3 + 16 \beta_2 \beta_3\right) \rho^2 + \\
+\left(-16 \beta_2^2 + 8 \beta_2^3 +16 \beta_3 - 8 \beta_2 \beta_3 \right) \rho^3 \Big),
\end{multline*}
\begin{multline*}
e_4 = 1-8 H \beta_1 + 16 H^2 \beta_1^2 + 8 H \beta_1 \beta_2 - 32 H^2 \beta_1^2 \beta_2 + 16 H^2 \beta_1^2 \beta_2^2 + \\
+ \Big( 8H \beta_1 -32 H^2 \beta_1^2 -16 \beta_2 +56 H \beta_1 \beta_2 + 64 H^2 \beta_1^2 \beta_2 + 8 \beta_2^2 -96 H \beta_1 \beta_2^2 - 32 H^2 \beta_1^2 \beta_2^2 + \\
+ 32 H \beta_1 \beta_2^3 + 8 \beta_3 + 32 H \beta_1 \beta_3 - 32 H \beta_1 \beta_2 \beta_3 \Big) \rho +\\
+ \Big( 16 H^2 \beta_1^2 + 16 \beta_2 - 128 H \beta_1 \beta_2 - 32 H^2 \beta_1^2 \beta_2 + 56 \beta_2^2 + 192 H \beta_1 \beta_2^2 + 16 H^2 \beta_1^2 \beta_2^2 - 64 \beta_2^3 - \\
- 64 H \beta_1 \beta_2^3 + 16 \beta_2^4 - 72 \beta_3 - 64 H \beta_1 \beta_3 + 64 \beta_2 \beta_3 + 64 H \beta_1 \beta_2 \beta_3 - 32 \beta_2^2 \beta_3 + 16 \beta_3^2 \Big) \rho^2 + \\
+32 \Big( 2 H \beta_1 \beta_2 - 4 \beta_2^2 -3 H \beta_1 \beta_2^2 + 4 \beta_2^3 + H \beta_1 \beta_2^3 - \beta_2^4 + 4 \beta_3 + \\
+H \beta_1 \beta_3 - 4 \beta_2 \beta_3 - H \beta_1 \beta_2 \beta_3 + 2 \beta_2^2 \beta_3 - \beta_3^2 \Big) \rho^3 
+ 16 \left(\beta_2^2 - \beta_3 ) (4 - 4 \beta_2 + \beta_2^2 - \beta_3 \right) \rho^4. 
\end{multline*}

We now show that $G\geqslant 0$, and $G= 0$ only if $\lambda = 2 \rho-1$. Following the same notation as Fuller~\cite{Fuller,Barros_Choi_2009} we consider the inner determinants $\Delta_7$, $\Delta_5$, and $\Delta_3$ for this quartic $G$ in \eqref{quartic_for_lambda}:
\begin{align*}
\Delta_7 &= 16777216 (\beta_2-1)^2 \beta_3^2 (\rho-1)^6 \rho^2 \,{T_1}^2(\rho, \alpha, H) \,T_2 (\rho, \alpha, H), \\
\Delta_5 &= 16384 \rho (\rho-1)^3 (\beta_2-1)\beta_3 \, (\kappa_0 \rho^2 + \kappa_1 \rho + \kappa_2 ),\\
\Delta_3 &= 64 (\rho-1) \Big((\beta_2-1) \left(H\beta_1 + \beta_2 (-1+2\rho)+2\rho \right) + \rho \beta_3 \Big),
\end{align*}
with $\kappa_i$ ($i=0,1,2$) dependent on $\rho$, $H$, and $\alpha$, and where the terms $T_1$ and $T_2$ are defined by: 
\begin{align*}
T_1 &= (H \beta_1-1) (\beta_2 -1) +\rho (\beta_2^2 -\beta_3 -1), \\
T_2 &= (\beta_2 -H\beta_1)^2 + 4\rho (H\beta_1 \beta_2 - \beta_2^2 + \beta_3) + 4 \rho^2 (\beta_2^2 - \beta_3). 
\end{align*}
By observing that $H\beta_1 \beta_2 - \beta_2^2 + \beta_3 = (H \alpha)^2 ( \coth \alpha \,\coth(H\alpha) - 1)$ ($>0$), we can conclude that $T_2>0$, and therefore $\Delta_7 \geqslant 0$. Extensive numerical tests show that $\Delta_3$ and $\Delta_5$ can never be simultaneously positive. So, whenever $\Delta_7>0$ we have $\Delta_3<0$, or $\Delta_5<0$, thus $G$ defined in \eqref{quartic_for_lambda} has 4 complex roots ({\it cf.} Theorem 4 in \cite{Fuller}). Since $\lambda \in \R$, we conclude that $G$ is positive. 
A simple way to arrive to same conclusion, without resorting to extensive numerical tests, is by observing that $\Delta_3 <0$ for relevant physical regimes. Namely, it follows from definition
%To see this, rewrite $\Delta_3$ as
%$$
%\Delta_3 = 64 (\rho-1) \Big((\beta_2-1) \left(H\beta_1 + \beta_2 (-1+2\rho)+2\rho \right) + \rho \beta_3 \Big),
%$$
that $\Delta_3<0$, when $\rho>0.5$. 
% As a matter of fact, it can be shown numerically that $\Delta_3$ is positive only for very small values of $\rho$ and $H$.

%When $\Delta_7=0$, then $\Delta_5<0$ and so $G$ has a double real root and two simple complex roots ({\it cf.} Theorem 6 in Fuller 1981).
When $\Delta_7=0$, then $T_1=0$. In this case, it can be shown that the quartic $G$ in \eqref{quartic_for_lambda} can be factorized as the product of two quadratics in $\lambda$. One of the terms can never vanish as it has only complex roots. The other can vanish provided $\lambda = 2\rho-1$. Nevertheless, the latter scenario corresponds to the case when the cubic equation \eqref{cubic_two} becomes degenerate, namely:
$$
\rho (1-\rho) \beta_3 \,v^2 +(\beta_2-1)^2 (1-\rho)^2 =0, 
$$
which has complex roots. Therefore, even though the polynomials in \eqref{cand_singular_pts} have a common root, it is necessarily a complex one. 

It remains one last case to be examined under which the resultant in \eqref{resultant_gen_case} can vanish. This corresponds to the case when 
$\lambda = -H\beta_1 -\rho(\beta_2-1)$. As it can be seen from \eqref{cubic_one} and \eqref{cubic_two}, the two cubic equations become degenerate, of degree 1 and 2, respectively. This can be used to assert that a common root exists for the polynomial equations only when $8 \rho (\rho-1)(\beta_2-1)^2 \beta_3=0$, thus corresponding to a degenerate case (see Appendix B).  

In summary, provided that $\alpha$ and $H$ are strictly positive, and $0<\rho<1$, there are no finite singular points for the family of curves defined by \eqref{curve_gen_case}.

\bigskip

\noindent{\bf Homogeneous case.} The analysis of the finite singular points obtained for \eqref{curve_gen_case} in the limit when $\rho\rightarrow 1$ is much simpler than the one presented above. The algebraic curve corresponding to our stability problem is described by the cubic curve \eqref{hom_case}, which can be written in homogeneous coordinates: 
$$
P(p,q,z)=P_3(p,q)+z^2 P_1(p,q)=0
$$
whose singular points (finite, or at infinity) are the solutions of $P_p=0$, $P_q=0$, $P_z=0$.
%{\it i.e.} they satisfy:
%\begin{eqnarray*}
%P_{3,p} + z^2 P_{1,p}=0,\\
%P_{3,q} + z^2 P_{1,q}=0,\\
%2 z \,P_1(p,q)=0.
%\end{eqnarray*}
%Singular points at infinity satisfy $z=0$, $P_{3,p}=0$, $P_{3,q}=0$, which implies that $P_3$ has a repeated factor. 
For finite singular points, we can assume without loss of generality that $z=1$ and solve:
\begin{eqnarray*}
P_{3,p}+P_{1,p}=0,\\
P_{3,q}+P_{1,q}=0,\\
P_1=0.
\end{eqnarray*}
If $\lambda=1$, then two singular points of the form $(0,\pm q_0,1)$ may or not exist, depending on the parameters considered. For this particular value of $\lambda$ the cubic becomes degenerate and splits up into a vertical line $p=0$ and a hyperbola, for which stability always holds. If $\lambda \neq 1$, then $P_1=0$ can be used to write $q=-(a_{1,0}/a_{0,1}) p$ that can be inserted into the first two equations to yield:
$$
(\lambda + H \beta_1 + \beta_2-1) (D_1 p^2 + E_1) =0,
$$
$$
D_2 p^2 + E_2 = 0.
$$
For the purpose of seeking finite singular points, one may discard the case when $\lambda + H \beta_1 + \beta_2-1=0$, since it only gives singularities at infinity. Other candidates to finite singular points must satisfy $D_1 E_2 =D_2 E_1$, {\it i.e,} 
$$ 
3 H \beta_3 (\lambda-1)^2 (\lambda + H \beta_1 + \beta_2-1) =0,
$$
which never occurs. Hence, unless $\lambda=1$, the cubic curve has no finite singular points.

\secte{Appendix E. Further details on the family of curves \eqref{curve_gen_case}}

We have seen throughout the text that not all transitions between curve configurations involve a change of the topological structure. If, in this process, different stability properties of the flow are obtained, then some caution is needed in the analysis, since we cannot rely on the singularities to detect such transitions. 
Consider the following example:
%such transitions cannot be detected by a study of the curve singularities.
%we cannot rely on the singularities to detect these transitions. 
%these will not be detected by a simple analysis of the curve singularities. 
%As such, a priori we cannot rely on the singularities to detect these transitions and some caution is needed. 
\begin{figure}[htbp]
\begin{center}
\includegraphics*[width=340pt]{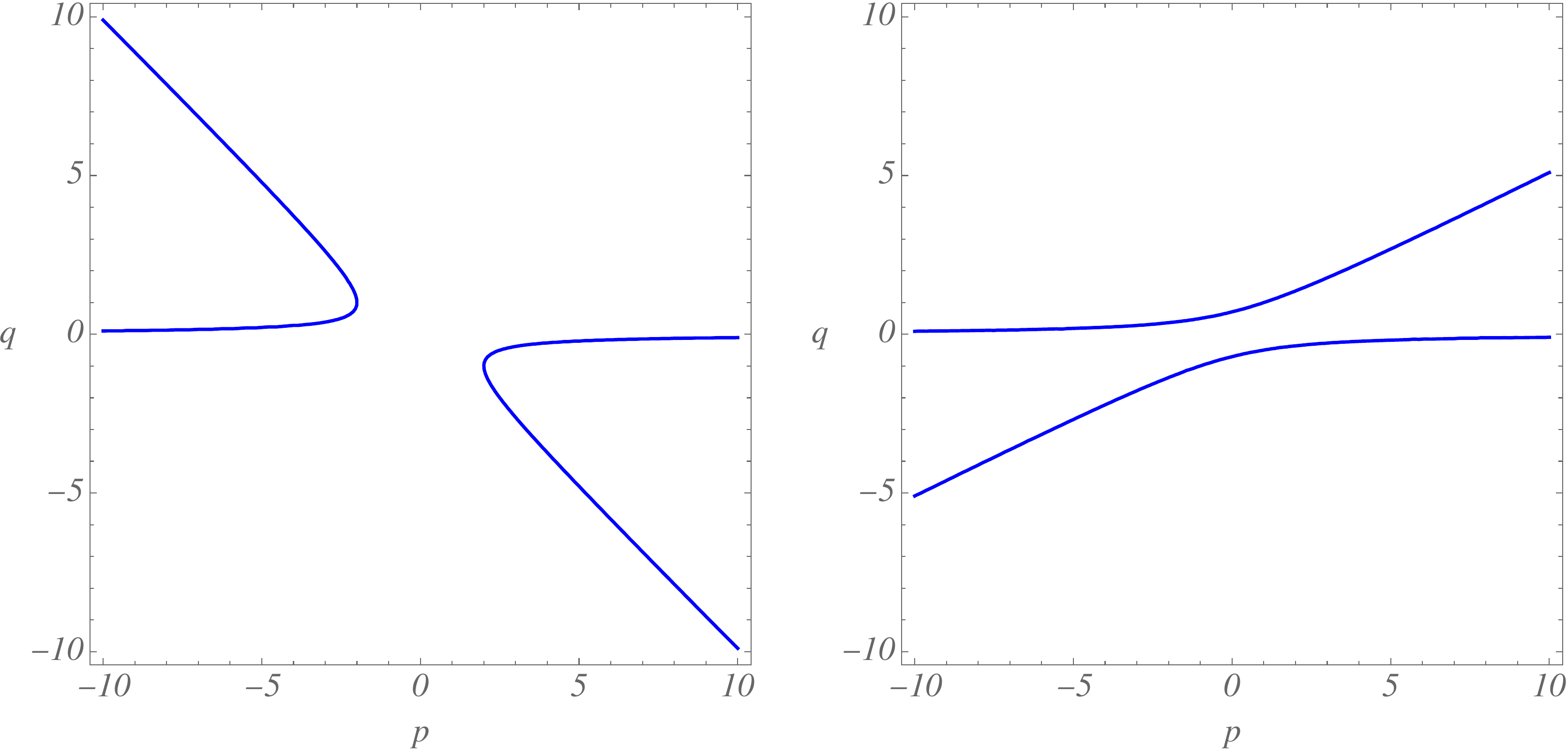}
\\
$(a)$ \hspace{4cm} $(b)$
\end{center}
\caption{Plots on the $(p, q)$-plane of the curve $P(p,q) \equiv q (ap+bq)-\epsilon=0$, with $\epsilon=-1$ and $a=1$. $(a)$ $b =1$. $(b)$ $b =-2$. The change between an unstable and a stable configuration occur without passing through a singularity. 
\label{example}
}
\end{figure}

%\begin{figure}[htbp]
%\begin{center}
%\includegraphics*[width=170pt]{figures/example_b_is_1}
%\qquad 
%\includegraphics*[width=170pt]{figures/example_b_is_Minus2}
%\\
%$(a)$ \hspace{4cm} $(b)$
%\end{center}
%\caption{Plots on the $(p, q)$-plane of the curve $P(p,q) \equiv q (ap+bq)-\epsilon=0$, with $\epsilon=-1$ and $a=1$. $(a)$ $b =1$. $(b)$ $b =-2$. The change between an unstable and a stable configuration occur without passing through a singularity. 
%\label{example}
%}
%\end{figure}
%

\begin{exe}
Consider the family of quadratic curves defined by $P(p,q) \equiv q (ap+bq)-\epsilon=0$, with real parameters $a$, $b$, and $\epsilon$. Suppose we want to determine conditions on the parameters for which any line with slope 1 always cuts the curve at two points. Let us denominate the curves with this property as being ``stable'', and ``unstable'' otherwise. By substituting $q=p+F$ into the curve equation, one obtains a quadratic equation for $p$, which has two distinct real roots provided the discriminant is positive. 
%(for all values of $F$). 
Since the discriminant is given by $a^2 F^2 + 4 (a+b)\epsilon$, clearly the desired condition is simply $(a+b)\epsilon >0$. 
\end{exe}
%This particular example serves as a cautionary tale for our work. 
In the example, for given fixed values of $a$ and $b$, the family of curves changes from a ``stable'' to an ``unstable'' configuration through a (finite) singularity ($\epsilon=0$). 
%The same applies if $a$ and $b$ vary in such a way that the sign of their sum is preserved. 
Notice, however, that we may also fix $\epsilon$ and $a$, and vary the values of $b$ to switch from a stable to an unstable configuration without passing through a singularity (see Fig.~\ref{example}), which should be used as a cautionary tale for our work. 

%However, if we fix $\epsilon$ instead, we may transition between two distinct configurations without passing through a singularity (see Fig.~\ref{example}).       

Although numerous configurations for the family of curves defined by Eq.~\eqref{curve_gen_case} can be found, based on its singularities, a full classification of the curve configurations will not be pursued here. Instead, we will focus on the short wave regime ($\alpha\gg 1$) and show the existence of solely two distinct configurations, albeit being topologically equivalent.  

%it is wise allowing for the possibility of other configurations to exist, resulting from changes of configurations that do not pass through singularities. 
%We will examine below some further details on the family of curves. 
For convenience, let us write the top degree form of $P(p,q)$ as
$$
P_4(p,q) = pq \left(a_{31} \,p^2 + a_{22} \,pq + a_{13} \,q^2\right),
$$
where the coefficients can be read from \eqref{P4_gen_case}. Our family of curves has a few things in common with the family with Eq.~\eqref{curve} described in Appendix A. 
First, curves with Eq.~\eqref{curve_gen_case} share the asymptotes $q=0$ and $p=0$, regardless of the physical parameters specified. Second, the way these curves approach the horizontal asymptote $q=0$ is the same as described for the Eq.~\eqref{curve}, since the curves admit the expansion $q= H/p + O(1/p^3)$, for large values of $p$.
% 
%When $p$ is large and positive, values of $q$ for the curve are greater than those on the asymptote. Hence the curve lies above the asymptote at that end. When $p$ is large and negative, the reverse situation occurs. 
Third, the hyperbola \eqref{hyperbola_from_degeneracy} is also tangent to the curves with Eq.~\eqref{curve_gen_case} at the $q$-intercepts, and is never crossed by the curve. 
Following the same steps as in Appendix A, we find the difference between the values of the ordinates for the hyperbola and the curve: 
$$
q_\text{hyp}-q_\text{curve} = - \frac{\rho H^2 \beta_3}{a_{31}} \frac{1}{p^3} + O(1/p^4).
$$
If $\lambda > -H$, or the value of $\alpha$ is large enough, then $a_{31}>0$. In such cases, for large positive values of $p$, values of $q$ for the hyperbola are less than those for the curve, and greater than those on the asymptote. When $p$ is large and negative, values of $q$ for the hyperbola are greater than those for the curve, and less than those on the asymptote.

Contrary to what was described for the curves with Eq.~\eqref{curve} the number of intercepts may vary. The $q$-intercepts are still the same, with coordinates $(0,\pm \sqrt{H/(\beta_2-1)})$, but curves with Eq.~\eqref{curve_gen_case} do not always cut the $p$-axis. For that to happen, $a_{31}$ must be positive, in which case we have the points $(\pm \sqrt{H(1-\rho)/a_{31}},0)$. %Remark that for large values of $\alpha$ this is always guaranteed.

\begin{figure}[htbp]
\begin{center}
\includegraphics*[width=340pt]{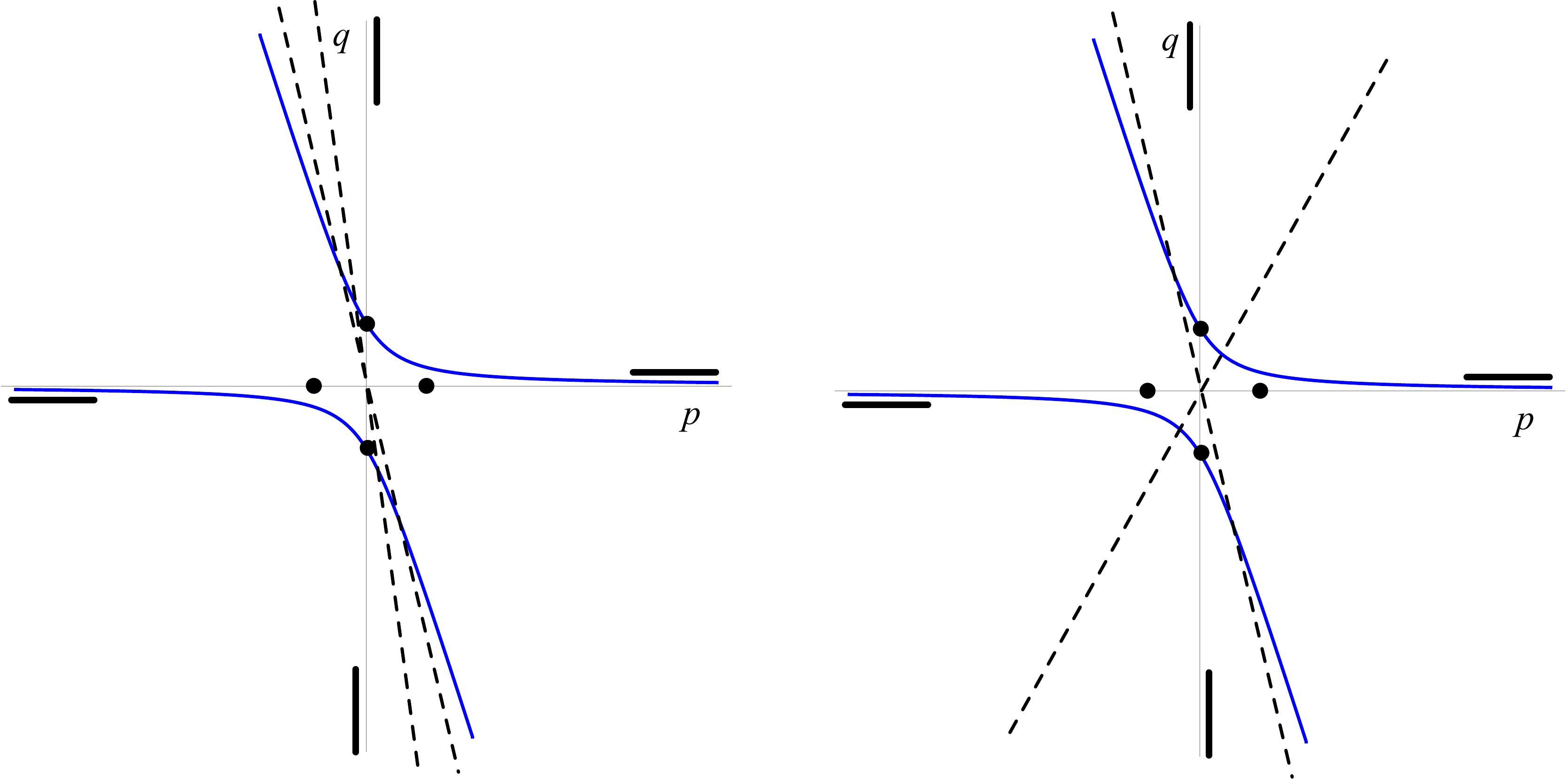}
\\
$(a)$ \hspace{4cm} $(b)$
\end{center}
\caption{Similar to Fig.~\ref{curve_sketch}. Schematic diagrams of the curve defined by Eq.~\eqref{curve_gen_case} for large values of $\alpha$. $(a)$ $\lambda<\rho$. $(b)$ $\lambda>\rho$. 
\label{configs_large_alpha}
}
\end{figure}

%\begin{figure}[htbp]
%\begin{center}
%\includegraphics*[width=170pt]{figures/rough_diagram_curve19_lambda_less_rho}
%\qquad 
%\includegraphics*[width=170pt]{figures/rough_diagram_curve19_lambda_greater_rho}
%\\
%$(a)$ \hspace{4cm} $(b)$
%\end{center}
%\caption{Similar to Fig.~\ref{curve_sketch}. Schematic diagrams of the curve defined by Eq.~\eqref{curve_gen_case} for large values of $\alpha$. $(a)$ $\lambda<\rho$. $(b)$ $\lambda>\rho$. 
%\label{configs_large_alpha}
%}
%\end{figure}

%We want to establish without the slightest doubt that the the flow is unstable, unless $\Omega_2=0$. 
We will focus on short waves ($\alpha\gg 1$) and show that these are always unstable. When the values of $\alpha$ are sufficiently large, two subregions in the parameter space can be distinguished, according to: $(i)$ $\lambda>\rho$; $(ii)$ $\lambda<\rho$. In the former case, coefficients $a_{31}$ and $a_{22}$ are positive, while $a_{13}<0$. Slopes for the two oblique asymptotes have opposite signs. In the latter case, all coefficients are positive, which means the two oblique asymptotes have negative slopes (see Fig.~\ref{configs_large_alpha}). In either case, the branches of the curve in second (fourth) quadrant are above (below) the hyperbola. 
%Similarly to what was described in Appendix A for the curve with Eq.~\eqref{curve}, the hyperbola defined by \eqref{hyperbola_from_degeneracy} is tangent to the curve with Eq.~\ref{curve_gen_case} at the $q$-intercepts, and is never crossed by the curve. 
When $\lambda<\rho$, the schematic diagram is similar to the one in Fig.~\ref{curve_sketch} and we find that in the first quadrant the asymptotes $q=0$ and $p=0$ are connected through the curve entirely within the quadrant. The oblique asymptotes in the second quadrant are connected through the curve entirely within the quadrant. 
On the other hand, when $\lambda>\rho$, we see a rearrangement of the curve branches. More precisely, in the first quadrant the asymptotes $q=0$ and the oblique one are connected through the curve entirely within the quadrant. In the second quadrant, the asymptotes $p=0$ and the oblique one are connected through the curve entirely within the quadrant. 
The behavior of the curve within the third and fourth quadrants follows by symmetry. 
%Let us consider the function $f(p,q) = P_4 + P_2 + P_0$ and examine its critical points. Clearly, the origin is a critical point regardless of the physical parameters considered. Moreover, symmetry about the origin implies that if $(p_0,q_0)$ are the coordinates of a critical point, then $(-p_0,-q_0)$ is also a critical point. We can then seek candidates for critical points along the lines of equations $q=m p$. For large values of $\alpha$  it can be shown that the values of $m$ are approximated by the roots of 
%$$
%\beta_2 (\rho-1) m^2 + H \beta_1 + \rho \beta_2 =0.
%$$
%Since the system resulting from $f_p=0$ and $f_q=0$ by substituting $q$ by $m p$ and excluding the trivial solution $p=0$ is composed by two equations of degree 2 for $p$, we then conclude that 5 critical points for the function $f$ are expected for large values of $\alpha$. The origin is a local maximum and the other four critical points are all saddle points. As such, 
%level sets around the origin are homeomorphic to circles. This is the explanation for the presence of an oval connecting the four intercepts.      

In the two cases, the hyperbola separates the inner component (a closed contour containing the four intercepts) from the outer components of the curve. As a consequence, at least for intermediate values of $F$, there will be instances when a line of equation $q=p+F$ intersects the curve only twice, and instability of the flow holds.

%A third case can obviously be identified ($\lambda=\rho$), allowing the transition between cases $(i)$ and $(ii)$. In such case, provided the values of $\alpha$ are sufficiently large, the curve has, in addition to the asymptote $q=0$, two parallel vertical asymptotes and one oblique asymptote (with negative slope), as in Fig.~\ref{lambda_is_rho}. Similarly to the cases $(i)$ and $(ii)$, instability holds for intermediate values of $F$. 

As a final remark, and relevant to \S~\ref{sec:uniform_vorticity}, it will pointed out that when $\lambda>\rho$ all the features discussed here for the case when $\alpha\gg 1$ hold for any strictly positive value of $\alpha$, even if arbitrarily small. As a consequence, one single configuration exists for the family of curves \eqref{curve_gen_case} when $\lambda>\rho$.

\end{document}